\documentclass[12pt, draftclsnofoot, onecolumn]{IEEEtran}
\usepackage[left=1in, right=1in, top=1in, bottom=1in]{geometry}

\IEEEoverridecommandlockouts
\newcommand{\subparagraph}{}
\usepackage{amsmath}
\usepackage{amssymb, amsfonts, amsthm}
\usepackage{algorithmic}
\usepackage{graphicx}
\usepackage{textcomp}
\usepackage[nospace,noadjust]{cite}
\usepackage{xcolor}
\usepackage{subcaption}
\usepackage{apptools}
\usepackage[justification=centering]{caption}
\usepackage{float}
\usepackage{bbm}
\def\BibTeX{{\rm B\kern-.05em{\sc i\kern-.025em b}\kern-.08em
		T\kern-.1667em\lower.7ex\hbox{E}\kern-.125emX}}
\usepackage{etoolbox}
\usepackage[pagestyles]{titlesec}
\titlespacing\section{0pt}{3pt}{1pt}%
\titlespacing\subsection{0pt}{2pt}{2pt}
\allowdisplaybreaks

\makeatletter
\def\endthebibliography{%
	\def\@noitemerr{\@latex@warning{Empty `thebibliography' environment}}%
	\endlist
}

\newenvironment{pf}[1][\proofname]{\par
	\pushQED{\qed}%
	\normalfont \topsep0\p@\relax
	\trivlist
	\item[\hskip\labelsep\itshape
	#1\@addpunct{.}]\ignorespaces
}{%
	\popQED\endtrivlist\@endpefalse
}

\makeatother
\newtheorem{theorem}{Theorem}[section]
\newtheorem{proposition}[theorem]{Proposition}
\newtheorem{corollary}{Corollary}[theorem]
\newtheorem{lemma}[theorem]{Lemma}

\theoremstyle{definition}
\newtheorem{definition}[theorem]{Definition}

\newcommand\independent{\protect\mathpalette{\protect\independenT}{\perp}}
\def\independenT#1#2{\mathrel{\rlap{$#1#2$}\mkern2mu{#1#2}}}

\begin{document}
	
	\title{How wireless queues benefit from motion: \\an analysis of the continuum between\\ zero and infinite mobility\\
		\thanks{This work was supported in part by the National Science Foundation under Grant
			No. NSF-CCF-1514275 and an award from the Simons Foundation (\#197982), both
			to the University of Texas at Austin and by the H2020 ERC award \#788851 to INRIA. N. Ramesan (nithinseyon@utexas.edu) is with the University of Texas at Austin, USA. F. Baccelli (francois.baccelli@austin.utexas.edu) is with the University of Texas at Austin, USA and INRIA, France. Part of this work appeared in the proceedings of WiOPT 2020 (\cite{ramesan2020wireless}).}
	}
	
	\author{Nithin S. Ramesan and Fran\c{c}ois Baccelli}
	\maketitle
	\vspace{-1.2cm}
	\begin{abstract}
		This paper considers the time evolution of a queue that is embedded in a Poisson point process of moving wireless interferers. The queue is driven by an external arrival process and is subject to a time-varying service process that is a function of the SINR that it sees. Static configurations of interferers result in an infinite queue workload with positive probability. In contrast, a generic stability condition is established for the queue in the case where interferers possess any non-zero mobility that results in displacements that are both independent across interferers and oblivious to interferer positions. The proof leverages the mixing property of the Poisson point process. The effect of an increase in mobility on queueing metrics is also studied. Convex ordering tools are used to establish that faster moving interferers result in a queue workload that is \textcolor{black}{smaller} for the increasing-convex stochastic order. As a corollary, mean workload and mean delay \textcolor{black}{decrease} as network mobility increases. \textcolor{black}{This stochastic ordering as a function of mobility is explained by establishing positive correlations between SINR level-crossing events at different time points, and by determining the autocorrelation function for interference and observing that it decreases with increasing mobility.} System behaviour is empirically analyzed using discrete-event simulation and the performance of various mobility models is evaluated using heavy-traffic approximations.
	\end{abstract}
	
	\begin{IEEEkeywords}
		Time-varying queues, queues in random environments, mobility, interference correlation, ad hoc networks, stochastic geometry, network dynamics, stochastic ordering, convex ordering, heavy-traffic approximations
	\end{IEEEkeywords}

	\section{Introduction}
 	With the rising popularity of V2V, V2X and drone networks, and more generally an increasingly widespread paradigm of mobility, the analysis of the effects of motion on wireless networks is increasingly relevant. Recent proposals (for example, \cite{sui2014deployment}) also advocate for the introduction of ``moving networks" - base stations mounted on vehicles to serve vehicular users. Current wireless networks also inherently involve a certain degree of motion - users in cellular networks are very often mobile, and there exist multiple use cases for ad hoc networks that involve mobile nodes. The effect of mobility on \textcolor{black}{certain metrics of performance of wireless networks has been extensively investigated. This paper, however, takes a new viewpoint on performance - that of queueing delay, or latency - and studies the effect of mobility on it.\\
	To study this metric, we relax the full buffer assumption that is widely used in wireless network analysis (for example, in \cite{gupta2000capacity} and throughout \cite{baccelli2009stochastic}). This rather ubiquitous assumption assumes that wireless transmitters always have packets to transmit, ignoring issues pertaining to traffic dynamics. While this assumption may be an attractive one due to the simplification it allows for, in reality most buffers behave as queues. The full-buffer assumption abstracts out the aforementioned metric of queueing delay - sacrificing knowledge of packet-level effects for tractability at the network level.}\\
	\textcolor{black}{We study a model in which a single receiver-transmitter pair of interest has a queue associated with it - a queue whose service process is a function of the SINR it experiences, and hence of the positions of interfering nodes around it, which are modelled as a planar Poisson point process. We assume that these points are moving with unconstrained but \textit{finite} velocity. To the best of our knowledge, the present paper is the first work to study the effect of finite and unconstrained mobility on the performance of queues in a wireless setting.} \textcolor{black}{This model has multiple applications - we mention three here. The first is an ad-hoc network setting, where the receiver and its associated transmitter can be thought of as a pair of devices, one transferring some data/files to the other, as in a device-to-device (D2D) network or a  wireless LAN (WLAN) V2X network. Interferers are other transmitter nodes sharing the medium. The second application is an uplink cellular network setting, where the receiver of interest is a base station and its associated transmitter is a mobile user equipment (UE). Interferers are other UEs sharing the medium. The third is the aforementioned ``moving network" downlink setting, where the receiver of interest is a UE, the associated transmitter is the UE's associated base station and interferers are other moving base stations. The interference experienced by the UE is then the inter-cell interference.}\\ Our model studies the continuum of interferer velocities in the interval $[0, \infty)$ and scenarios where the point processes seen at two different time instants at finite distance are correlated. The understanding of this correlation and its effects is the main aim of the present paper. The correlations we see in interferer configurations will induce time-correlations in interference and hence the service of the queue - correlations that make analysis of the behaviour of the queue non-trivial.\\ \textcolor{black}{Our model is novel and challenging to analyze because it deals with two types of coupled network dynamics - (i) the motion of interferers changes their locations and hence triggers time evolution of the interference and service rate seen by the queue, and (ii) the queueing dynamics induce changes in the length of queue buffer and hence has implications on metrics like delay. These processes are coupled because the time evolution of the queue depends directly on the time evolution of the positions of interferers around it. In summary, the interplay of stochastic geometry and queueing is a challenging problem - the former deals with space averages, and the latter with time averages. We will show that mobility unifies these disjoint paradigms.}
	\subsection{Related Work}
	\textcolor{black}{This paper deals with the effect of mobility on a queue in a wireless network. To the best of our knowledge, while there has been previous work on the effect of mobility on wireless networks and on the evolution of queues in a wireless environment, this is the first work to consider the coupling of both mobility and queueing together. We first review these 2 areas since our work lies at their intersection, and then review other relevant literature.\\
	The effects of mobility on the performance of wireless networks have been studied extensively, beginning with the seminal work \cite{grossglauser2002mobility} (for a more thorough literature review, see \cite{ying2008optimal} \& \cite{bao2015stochastic} and references therein). The vast majority of works in this area (for example, \cite{lin2006degenerate, sharma2007delay, neely2005capacity, li2010throughput}) are concerned with the concepts of transport capacity or throughput capacity (as first proposed in \cite{gupta2000capacity}) or multihop delay (as in \cite{lin2006degenerate}) rather than queueing capacity and delay, as is our focus. These works also do not consider the traffic dynamics of the systems they study by assuming a backlogged/full buffer for tractability, an assumption that we relax in our work. As we mentioned in the introduction, we consider mobility models where wireless transmitters move with finite velocities. In contrast, most analytical work on the effect of mobility on wireless networks assumes i.i.d. redistribution of nodes in time slots (for example, see \cite{grossglauser2002mobility} and \cite{haenggi2010local}). Despite these differences, a common thread between this body of work and our results is the important observation that mobility improves performance. \\
	There has been a recent body of work that seeks to study problems that combine queueing and wireless networks (see \cite{zhong2016towards} for a brief survey of this class of problems). A number of these works seek to establish necessary and sufficient conditions for stability using stochastic dominance ideas (e.g., \cite{zhong2016stability, yang2020optimizing}). Other works make independence or mean-field assumptions (e.g., \cite{yang2018sir, yang2021understanding, chisci2017scalability, chisci2019uncoordinated}). The work closest to ours in terms of model (\cite{stamatiou2010random}) assumes an infinite mobility model, in constrast to the finite mobility models we assume. The effect of all these assumptions is to decorrelate interference across time slots, removing any time correlations. In our work, we analyze the system without making any independence or mean-field assumptions. Indeed, it is the time correlations so preserved that lead to the stochastic ordering results we establish later in the paper. Another line of thought treats the entire wireless network as a queue, dealing with the birth and death of transmitters and receivers (\cite{sankararaman2017spatial} and \cite{alammouri2019stability}) - but does not consider the effect mobility will have on the network dynamics.\\
	The preservation of time correlations that we mention above results in a queue that sees time-varying service rates. There is an existing line of work on time-varying queues modelled using non-homogenous Poisson processes as arrival and service processes (e.g., \cite{massey2002analysis}, \cite{feldman2008staffing}) - our service process is not amenable to this framework because it is not a Poisson process. Modelling more general arrival and service processes has been accomplished by considering Markov-modulated queues - where arrival and service rates are functions of the state of a continuous-time Markov chain with finite state space (the classical Gilbert model of queueing theory is a special case where the CTMC has two states). Under this framework, computational methods have been proposed to calculate the steady state probabilities (\cite{neuts1977m}, \cite{kao1989m}). To attempt to use such an approach, our model would have to be made Markovian by assuming knowledge of the interferer point process instead of SINR. Its evolution would then have to be modelled as a CTMC with infinite dimensional state space (one coordinate for every interferer) - rendering this method intractable.\\
	Finally, we review the concept of local delay, introduced in \cite{baccelli2010new}. The local delay in a network is the random number of time slots required by a node to transmit a packet successfully to its intended (single-hop) destination. This can instead be thought of as the number of time slots required for the packet located at the head of a queue such as the one we consider to be successfully transmitted. While local delay implicitly excludes queueing-level effects, it is the performance metric closest in concept to the queueing metrics of delay/buffer length that we study here. The mean local delay is computed as a spatial average of the local delays of all nodes in a network, where by contrast, our queueing system involves temporal evolution and averaging to obtain queue statistics. The infinite mobility model is applied to local delay in \cite{baccelli2010new} and \cite{haenggi2010local}. Another work that considers finite mobility (\cite{gong2013local}) uses a constrained mobility model where locations of nodes are independent excursions from a fixed home-location point process, allowing for tractability. In contrast, we consider an unconstrained mobility model. Finally, a body of work related to local delay is on the meta distribution (e.g., \cite{haenggi2015meta}) - the distribution of the conditional probabilities of SINR level-crossing events (such events are discussed in Section \ref{sec:correlation}), conditioned on network geometry.}\\
	\subsection{Contributions}
	\textcolor{black}{The two main questions that we answer are queueing-theoretic in nature: 1) when, if at all, is the queue stable, and 2) how do queue statistics change with varying degrees of mobility? The first question considers a notion of capacity that arises from queueing theory, i.e., the maximum input data rate that a buffer can support whilst remaining stable. In answering this question, we show that in the absence of mobility in the network, universal stability guarantees do not exist (Proposition \ref{prop:instability}). This result stems from the observation that a static network is exactly that - static and unchanging, and hence a queue will see the same interference at all times - dooming it to stability or instability forever, subject to the caprice of the interferer configuration around it. But all is not lost - we show that the introduction of \textit{any} non-zero mobility results in enough change in the system that we can guarantee stability of our queue - independently of the positive degree of mobility in the system (Theorem \ref{queue_stability}). We do so by establishing that the point process of interferers (Theorem \ref{pp_strongly_mixing}) and hence the interference process (Corollary \ref{interference_mixing}) seen by the queue are strongly mixing and ergodic, allowing us to use Loynes' well known result on queue stability (\cite{loynes1962stability}). The second question considers the queueing delay experienced by packets in this queue, and the effects of mobility on this metric. Considering a continuum of velocities, we find that increasing motion in the network results in an accelerated rate of variation in the service rate of the queue - which in turns induces an averaging effect that provides increasingly smooth, or less variable service (Lemma \ref{cvx_ordering_service}). The continuum of velocities in $(0, \infty)$ is hence shown to induce a continuum of stochastically ordered queue workloads (Theorem \ref{queue_order_thm}). The overall effect is that of decreasing mean queue workload and mean delay (Corollaries \ref{mean_workload_ordering} and \ref{mean_delay_ordering}). We use the theory of stochastic ordering to obtain these results. We also explain the queue ordering results we obtain by establishing results on the correlations that exist in moving wireless networks (Theorem \ref{positive_correlation_crossing} and Lemma \ref{positive_correlation_coefficient}). We supplement our theoretical results with simulations that study the effects of different mobility models and heavy-traffic approximations that we use to accurately estimate the mean queue workload.}\\
	The rest of this paper is structured as follows: in Section \ref{sec:model}, we present our queueing and wireless network models. In Section \ref{sec:mobility_guarantees_stability} we answer the first question we asked (concerning stability) and in Section \ref{sec:mobility_orders_workloads} we answer the second question (concerning the effect of increasing mobility on queue statistics). Section \ref{sec:correlation} formalizes some of the intuition about correlations in our model, which we use to explain the results we presented in Section \ref{sec:mobility_orders_workloads}. Section \ref{sec:interacting} discusses how our insights can be applied to the more general interacting queues setting, and Section \ref{sec:simulation} concludes with simulation and computational studies.
	\section{Model}
	\label{sec:model}
	We assume a continuous time model.  We assume that interferers (which are for all intents and purposes, transmitters) are initially distributed as a homogenous Poisson point process $\Phi_0$ of intensity $\Lambda$ on $\mathbb{R}^2$. Denote by $\Phi_t$ the point process of interferers at time $t$. We assume that a receiver of interest is fixed at the origin, with an associated transmitter transmitting to the receiver at unit power. The location of the transmitter is given by a point chosen uniformly at random on the circle centered at the origin and with radius $R$. We assume that all interferers transmit with unit power at all times. This assumption is for notational simplicity - our framework allows for a more general model that incorporates power control and/or contention protocols that are functions of $\Phi_t$ (we will briefly explain how the results we establish still hold in their respective sections). A general isotropic path-loss function $l(.): \mathbb{R}^+ \rightarrow \mathbb{R}^+ \cup \{0\}$ is assumed, with the following properties: 1) $\lim_{ r \rightarrow \infty} l(r) = 0$ (in other words, we assume that an interferer that is infinitely far away does not have any significant effect at the origin), and 2) $\int rl(r) dr < \infty$ (which ensures that the interference shot-noise at any point is finite a.s.). Assume that $F_j^0(t)$ (the fading between interferer $j$ and the origin) is a general stationary stochastic process with continuous sample paths for all $j$, and that $F_j^0(t)$  is independent across interferers, i.e., $F_j^0(t) \independent F_i^0(t), i \neq j \text{ and } \forall t$. \textcolor{black}{We assume that $F_j^0(t)$ is a strongly mixing process for all $j$ (formally defined in Def. \ref{strong_mixing_def}), which essentially means that $F_j^0(t)$ and $F_j^0(t+s)$ are independent for sufficiently large $s$.}  
	Let $\gamma$ denote the thermal noise power at the receiver. Let $S_0(t)$ be the received signal power at the origin, and $I_0(t)$ be the interference seen at the origin. The SINR seen by the receiver at the origin at time $t$ can then be described by:
	\begin{align}
	\label{sinr_origin}
	\text{SINR}_0(t) &= \frac{S_0(t)}{I_0(t) + \gamma}
	= \frac{l(R)F_0^0(t)}{\sum_{x_j \in \Phi_t}l(|x_j|)F_j^0(t) + \gamma}.
	\end{align}
	\textcolor{black}{We assume here that interference is treated as noise - this is a common assumption in stochastic geometry literature (see the books \cite{baccelli2001coverage, baccelli2009stochastic2, blaszczyszyn2018stochastic} and the review/tutorial articles \cite{haenggi2009stochastic, andrews2016primer})}.\\Our results also hold for the generalized model where the transmitter associated with the receiver moves around (i.e., if the distance $R$ is not constant but instead varies over times as $R(t)$), as long as we assume that the transmitter employs a form of channel inversion and transmits with power proportional to $l(R(t))^{-1}$. This has the effect of ensuring constant received signal power up to fading variables. We can also analyze the case where the receiver is moving (relevant in the ad-hoc setting) - we will briefly address this in Sections \ref{sec:mobility_guarantees_stability} and \ref{sec:mobility_orders_workloads}.\\
	\textcolor{black}{In all versions of the model, the primary object of interest is the behaviour of $I_0(t)$, which is the value of a time-space interference shot noise field at the origin (see \cite{baccelli2009stochastic}, Section 2.3.4), with the positions of interferers changing over time. This time-variation is governed by our mobility model. We impose two conditions on any mobility model we consider: (i) the (possibly random) trajectories that any two different interferers follow in any interval of time must be independent of each other and (ii) these trajectories should not be a function of the locations of interferers in space - i.e., the motion of any one interferer should be independent of all other interferers, and also independent of that interferer's location in space.} We will now present three examples of mobility models that satisfy these independence conditions. \\\textit{Random Direction Model (RD): }Assume that at the beginning of time, every interferer $i$ samples an angle $\theta_i$ uniformly at random from the interval $[0, 2\pi]$. Each interferer will then move with constant velocity $v \in [0, \infty)$ along this angle - hence, in a time interval $\Delta t$, interferer $i$ will be displaced by $(v\Delta t \cos \theta_i, v\Delta t \sin \theta_i)$.\\
	\textit{Random Waypoint Model (RWP): }We consider a version of the random waypoint model where every interferer $i$ moves with constant velocity $v$ and angle $\theta_i$ sampled uniformly from $[0, 2\pi]$ for a deterministic amount of time $\Delta p$. At the end of every such time interval, every node resamples its angle of motion and continues moving.\\
	\textit{Brownian Motion (BW): } We assume that every interferer moves according to an independent 2-D Wiener process, with the variance of each one-dimensional Wiener process being $\sigma^2$. While the mean displacement of an interferer in any time interval will be zero, the Euclidean norm of displacement is given by a Rayleigh distribution with parameter $\sigma$. The mean of such a distribution is $\sigma\sqrt{\frac{\pi}{2}}$. In a time interval $\Delta$, we expect that an interferer moving with velocity $v$ will be displaced by $v\Delta$. Hence, during simulations in Section \ref{sec:simulation}, the variance of the 1-D Wiener processes is chosen such that $\sigma\sqrt{\frac{\pi}{2}} = v\Delta$, where $\Delta$ is the size of the small time interval in which the system evolves in simulations. This can easily be extended to more general diffusion processes.\\
	 By the displacement theorem for PPPs (\cite{baccelli2009stochastic}, Theorem 1.3.9), it follows that for all three models, $\Phi_t$ is a homogenous PPP for all $t$. Note that all queue comparison results presented in Section \ref{sec:mobility_orders_workloads} hold for these models and for more general mobility models that satisfy the aforementioned independence conditions.
	\subsection{Queueing model}
	We assume that the transmitter-receiver pair of interest is equipped with a queue. Data arriving to the transmitter enters the queue, and must be transmitted to the receiver, upon which the data leaves the queue. To describe the queueing dynamics, we assume that a discrete time scale is overlaid on our continuous time scale, with time slots of duration $\delta $. We consider a single-server, infinite buffer queue. We assume that packets, each of unit size, arrive to the queue according to an external discrete time stationary and ergodic arrival process, $A(n)$, which is independent of the PPP and the interferer motion processes. In other words, the queue workload $W_n$ (amount of data waiting in the queue for transmission) at time $n\delta$ increases by $A(n)$. Note that we can approximate the case of a continuous time arrival process by making $\delta $ arbitrarily small. The departure process associated with this queue depends on the interference time-series that the receiver sees. We assume a continuous-time random service process, $s(t)$, \textcolor{black}{upon which we place the three restrictions that $s(t)$ must be a measurable function of $\text{SINR}_0(t)$, that $s(t)$ must be integrable and that $s(t)=0$ when $\text{SINR}_0(t)=0$}. Examples of possible service processes include: a) Shannon rate: $s(t) = \log_2(1 + \text{SINR}_0(t))$ and b) truncated Shannon rate: $s(t) = \log_2(1 + \text{SINR}_0(t))\mathbbm{1}[\text{SINR}_0(t)>T]$. The former assumes adaptive modulation and coding based on SINR, whereas the latter assumes the same as long as the SINR is above a certain threshold $T$. It is easy to see that $s(t)$ is a stationary process. Such a model is in part a fluid queueing model - while we are implicitly assuming that data arrives as packets, we assume that it is transmitted continuously as a data stream. The cumulative service that the queue sees in the time interval $(t_1, t_2]$ is then denoted by $S(t_1, t_2) = \int_{t_1}^{t_2} s(t) dt.$
	Since the arrival process is in discrete time, the queue also evolves in discrete time according to the Lindley recursion:
	\begin{align}
	\label{lindley_recursion}
	W_{n+1} = \left(W_{n} + A(n)  - V(n) \right)^+,
	\end{align}
	where $W_n$ is the workload of the queue at time slot $n$ or at time $n \delta $, $V(n) = S(n\delta , (n+1)\delta )$ and $x^+ = \max(0, x)$.\\
	This queue is a G/G/1/$\infty$ queue - where we emphasize that service to the queue is stationary, generally distributed, and \textit{dependent}, namely \textit{not independent over time slots} - seeing a large service rate in one time slot will increase the likelihood of a large service rate in the next time slot, and similarly for small service rates. This fact stems from the finite mobility we consider in our model, which induces positive time-correlations in interferer configurations and hence in $s(t)$. We will make the idea of this correlation more precise in Section \ref{sec:correlation}. This lack of independence makes an already non-trivial problem even more complex.
	\begin{table}
		\centering
		\begin{tabular}{ ||c|c|| } 
			\hline
			$\Phi_t$ & Point process of interferers at time $t$  \\ 
			\hline
			$\Lambda$ & Intensity of interferer point process  \\ 
			\hline
			$R$ & Distance between receiver at origin and its transmitter \\
			\hline
			$l(r)$ & Path-loss function\\
			\hline
			$F^0_j(t)$ & Fading between interferer $j$ and the origin\\
			\hline
			$\gamma$ & Thermal noise\\
			\hline
			$I_0(t)$ & Interference shot-noise at the origin\\
			\hline
			$v$ & Velocity of interferers\\
			\hline
			$\delta$ & Time-scale at which queue evolves\\
			\hline
			$s(t)$ & Service process for queue\\
			\hline
			$A(n)$ & Arrivals to queue in time-slot $n$\\
			\hline
			$V(n)$ & Service to queue in time-slot $n$\\
			\hline
			$W_n$ & Queue workload at time-slot $n$\\
			\hline
		\end{tabular}
		\label{table:notation}
		\caption{\textcolor{black}{Table of notation.}}
	\end{table}
	\section{Mobility Allows  For Stability Guarantees}
	\label{sec:mobility_guarantees_stability}
	In this section, we will contrast the effects of a static configuration of inteferers on the evolution of our queue with the effects of a configuration of interferers that possess a non-zero degree of mobility, modelled using the random direction model. The mixing results that we show in this section will extend to the other two models we detailed and to other mobility models as well, but we do not include those proofs here. 
	For this section, denote by $(\Omega, \mathcal{F}, \mathbb{P})$ the probability space upon which the initial PPP of interferers are defined. Consider a static version of our model, with $v=0$. The initial configuration of interferers will be fixed for all times, i.e., $\Phi_0 = \Phi_t \text{ } \forall t$. The service process $V(n)$ is clearly stationary, but it is not ergodic - time averages of $V(n)$ will not be the same as ensemble averages over point process realizations. Assume that the arrival process $A(n)$ has rate $\lambda = \mathbb{E}[A(n)]$. \\Since the queue workload is driven in part by a non-ergodic sequence $V(n)$, Loynes' Theorem does not apply for this queue (see Section 2.1, \cite{baccelli2013elements} for the necessary conditions to apply Loynes' Theorem) and stability is not guaranteed for any $\lambda$. Indeed, we will now see that instances of \textit{instability} are guaranteed for all $\lambda$. Note that there exists a subset  $\Omega_U$ of $\Omega$ that has non-zero measure and that corresponds to point process realizations that have interferers clustered near the origin, resulting in the instability of the queue. The other side of this coin is also true - there will be interferer configurations (for $\omega \in \Omega_U^C$) where all interferers are sufficiently far away and where the queue will experience high rates of service (see Fig. \ref{fig:good} and \ref{fig:bad}). The stationary queue workload for $\omega \in \Omega_U^C$ will be finite, and the queue workload for $\omega \in \Omega_U$ will be infinite. Characterizing $\Omega_U$ is straightforward - $\Omega_U = \{\omega \in \Omega : \lambda > \delta \mathbb{E}[s(t)|\Phi(\omega)]  \}$, with $\Phi(\omega) = \Phi_0$. \\As an example, let $s(t)$ be the Shannon rate, $s(t) = \log_2[1+\text{SINR}_0(t)]$. Then, 
	\begin{align}
	\label{eq:example_shannon}
	\Omega_U = \left\{\omega \in \Omega : \lambda > \delta\mathbb{E}\left[\log_2\left(1 + \frac{l(R)F_0^0(0)}{\sum_{x_j \in \Phi(\omega)}l(|x_j|)F_j^0(0) + \gamma}\right)\Big| \Phi(\omega) \right] \right\}.
	\end{align}
	It is sufficient to consider the expected value of rate (conditioned on $\Phi$ and with expectations taken over the fading process values at time $t=0$) because i) from our assumptions on the fading processes, we can conclude that they are strongly mixing (see Def. \ref{strong_mixing_def}) and are hence ergodic and ii) \textit{conditionally upon} $\Phi$, the service process of the queue is ergodic and iii) Loynes' Theorem is concerned only with the comparison of average input and output rates of the queue. The set $\Omega_U$ in (\ref{eq:example_shannon}) is non-trivial because for any $\lambda$, we can find a non-trivial collection of interferer positions (eq., point process realizations) such that interferers are sufficiently closely clustered around the origin (as in Fig. \ref{fig:bad}), resulting in high interference and hence an average Shannon rate (conditioned on interferer positions) that is lower than $\lambda$.\\
	\textcolor{black}{We can make the computation of $\mathbb{P}(\Omega_U)$ explicit in certain cases by noting that (for a general service process $s(t)$), $\mathbb{P}(\Omega_U) = \mathbb{P}\left[s(0) < \frac{\lambda}{\delta}\right]$, which we can compute using standard techniques. If $s(t) = \log_2[1+\text{SINR}_0(t)]$ (like in (\ref{eq:example_shannon})) and assuming Rayleigh fades with mean $1/\mu$, then we have (see Proposition 16.2.2 in \cite{baccelli2009stochastic2}):
\begin{equation}
	\mathbb{P}(\Omega_U) = 1 - e^{-\gamma \mu T/l(R)} \exp\left(-2\pi\Lambda \int_0^\infty \frac{u}{1 + l(R)/(Tl(u))}du\right)\text{, where $T=e^\frac{\lambda}{\delta} - 1$}.
\end{equation}
We can characterize $\mathbb{P}(\Omega_U)$ even when $\Phi_0$ is not a Poisson process, as long as the Laplace transform of the shot-noise process of $\Phi_0$ is known (e.g., if $\Phi_0$ is a determinantal point process, we can use Lemma 9 in \cite{li2015statistical} to compute $\mathbb{P}(\Omega_U)$, which will again be positive).}
		 	\begin{figure}[h]
		 		\label{good_and_bad}
		 	\begin{center}
		 		\begin{subfigure}[h]{0.45\columnwidth}
		 			\includegraphics[width=\linewidth]{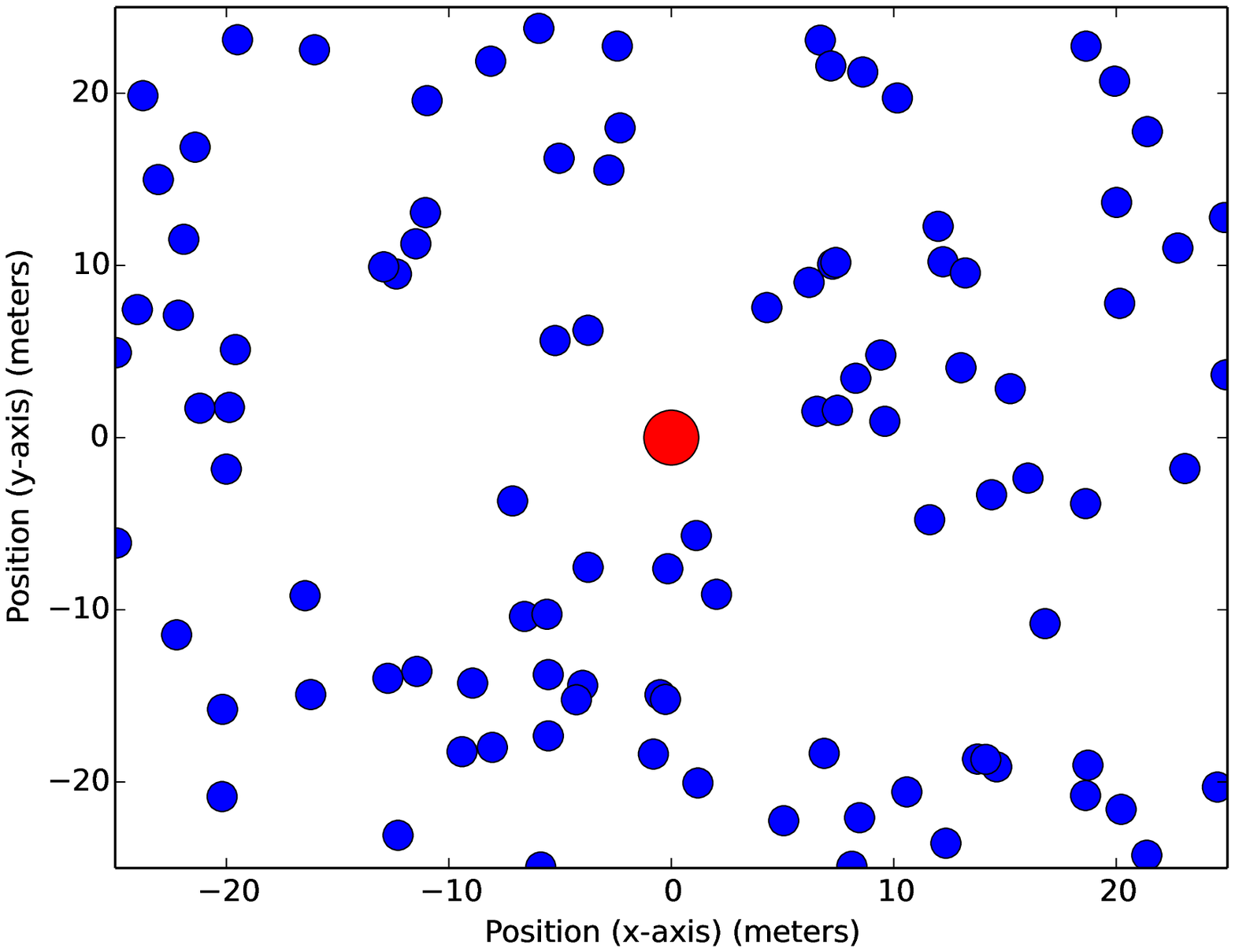}
		 			\caption{$\omega \in \Omega_U^C$}
		 			\label{fig:good}
		 		\end{subfigure}
		 		\begin{subfigure}[h]{0.45\columnwidth}
		 			\includegraphics[width=\linewidth]{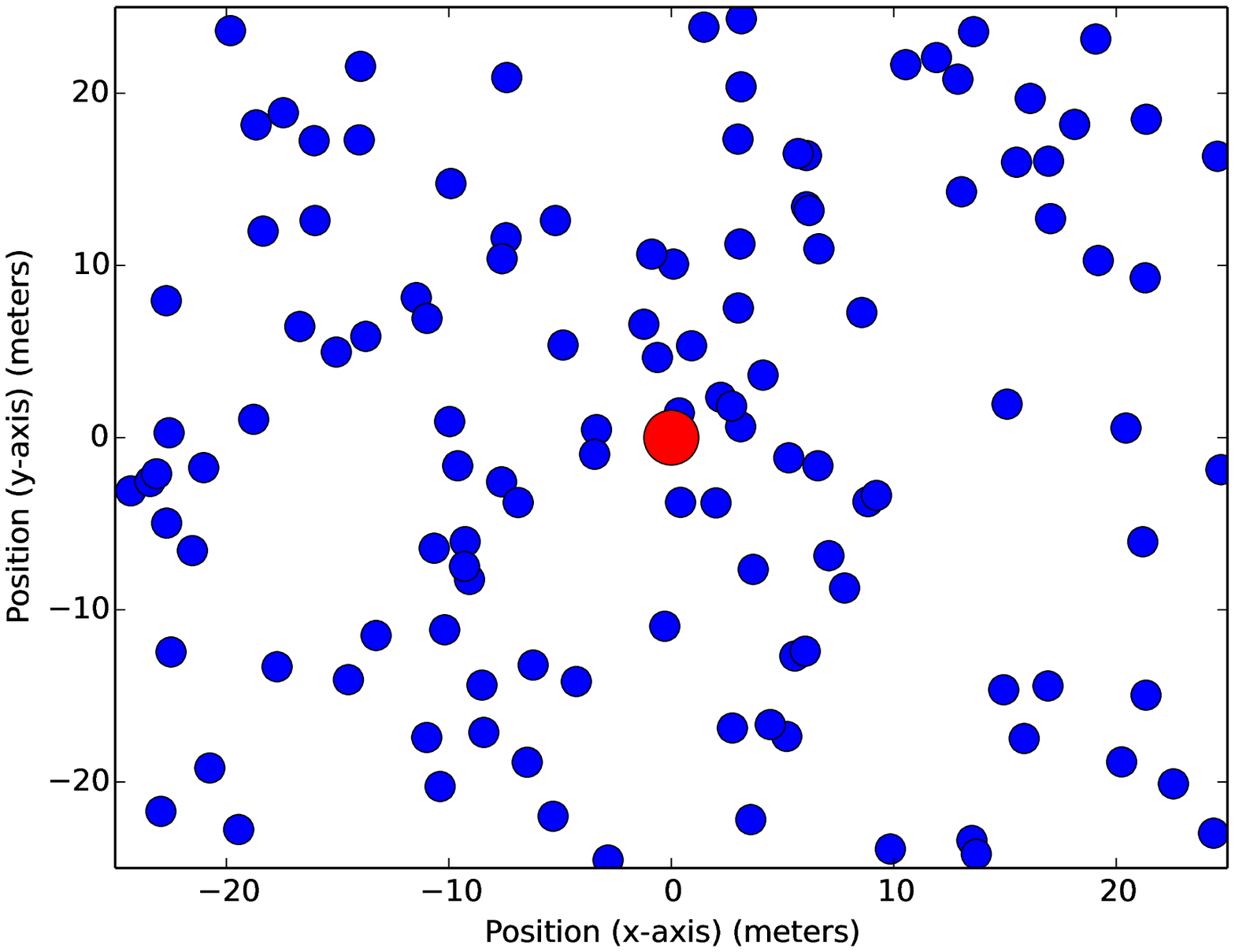}
		 			\caption{$\omega \in \Omega_U$}
		 			\label{fig:bad}
		 		\end{subfigure}
		 	\end{center}
	 	\caption{Examples of ``good" and ``bad" configurations of interferers. The red circle is the receiver at the origin, and blue circles are interferers.}
		 \end{figure}
		 \begin{proposition}
		 	\label{prop:instability}
		 	For a static field of interferers, the stationary queue workload will be infinite with positive probability.
		 \end{proposition}
	 \begin{proof}
	 	Follows from noting that the stationary queue workload is infinite on $\Omega_U$, which is a set with non-zero measure on the space $(\Omega, \mathcal{F}, \mathbb{P})$.
	 \end{proof}
 \textcolor{black}{As we noted above, $\Omega_U$ can have non-zero measure for more general point processes than just the Poisson process, in which case Proposition \ref{prop:instability} will still hold.}
	  In words, Proposition \ref{prop:instability} says that there is no arrival rate for which we can guarantee sample path stability. From a system design perspective, this is not ideal. In order to guarantee that for a set of input data rates our wireless buffer is \textit{always} stable, one solution as we will show now is to introduce mobility into the network.
	  We first present some useful definitions.\\
	  Consider a probability space $(\Omega, \mathcal{F}, P).$ For any two $\text{sub }\sigma$-algebras $\mathcal{A}$ and $\mathcal{B}$ of $\mathcal{F}$, define:
	  \begin{equation*}
	  \alpha(\mathcal{A}, \mathcal{B}) = \sup |P(A \cap B) - P(A)P(B)|, A\in\mathcal{A}, B\in\mathcal{B}
	  \end{equation*}
	  Let $X = \{X_k\}$ be a random process. Define
	  \begin{equation*}
	  \mathcal{F}_j^l = \sigma(X_k, j \leq k \leq l).
	  \end{equation*}
	  Then define 
	  \begin{equation*}
	  \alpha(n) = \sup_j \alpha(\mathcal{F}_{-\infty}^j, \mathcal{F}_{j+n}^\infty).
	  \end{equation*}
	  \begin{definition}
	  	\label{strong_mixing_def}
	  	\textbf{\cite{bradley2005basic} Strong mixing for stationary processes:} 
	  	If $X$ is a stationary process, then $X$ is said to be strongly mixing if $\alpha(n) \rightarrow 0$ as $n \rightarrow \infty$, where $\alpha(n) = \alpha(\mathcal{F}_{-\infty}^0, \mathcal{F}_{n}^\infty)$.
	  \end{definition}
	  It trivially follows that $F_j^0(t)$ is strongly mixing, since $\alpha(n)=0$ for $n \geq \tau^c_j$ (from the coherence time assumptions we made about the channel).
	  \begin{definition}
	  	\label{strong_mixing_markov}
	  	\textbf{\cite{bradley2005basic} Strong mixing for stationary Markov processes:} If $X$ is a Markov process, then $X$ is said to be strongly mixing if $\alpha(n) \rightarrow 0$ as $n \rightarrow \infty$, where $\alpha(n) = \alpha(\sigma(X_0), \sigma(X_{n}))$.
	  \end{definition}
	  Denote by $\alpha(X, n)$ the value of $\alpha(n)$ for a stochastic process $X = \{X_t\}$. We then have the following lemma, which we will use in the sequel.
	  \begin{lemma}
	  	\label{mixing_condition_compositions}
	  	(\cite{bradley2005basic}, Thm 5.2) Suppose that for every $n=1, 2, 3, ...$, $X^{(n)} = \{X^{(n)}_t\}$ is a stochastic process taking values in the set $S_n$, and that these processes are independent of each other. Suppose that for every $t$, $h_t: S_1 \times S_2 \times S_3 ... \rightarrow \mathbb{R}$ is a Borel function. Define $X = \{X_t\}$ by $X_t = h_t(X^{(1)}_t, X^{(2)}_t, ...)$. Then for all $m$, 
	  	\begin{equation*}
	  	\alpha(X, m) \leq \sum_{n=1}^\infty \alpha(X^{(n)}, m).
	  	\end{equation*}
	  \end{lemma}
In the presence of mobility in the network, we can state the following theorem.
\begin{theorem}
	\label{pp_strongly_mixing}
	If $v>0$, the process $\{\Phi_t\}$ is strongly mixing.
\end{theorem}
\begin{proof}
	Note first that the process $\{\Phi_t\}$ is Markov and stationary. Denote by $s$ the random displacement that transforms the point process $\Phi_t$ to $\Phi_{t+1}$ for any time $t$. For the sake of the proof, we assume the random direction mobility model, but this is not necessary, and we address the question of generalizing to other mobility models at the end of this section. This random displacement $s$ is characterized by a law $G(.)$, which in our case uniformly samples angles for every point at time $t=0$ from the interval $[0, 2\pi]$, while holding speed constant. Hence, $s$ takes values in the boundary of the 2-D ball centered at the origin with radius $v$, $\partial B(0, v)$. From the definition of our mobility model, we know that the displacement that transforms the point process $\Phi_t$ to $\Phi_{t+n}$ is given by $ns$. To prove the theorem, we must show that in the limit $n\rightarrow \infty$, the point processes $\Phi_t$ and $\Phi_{t+n}$ are asymptotically independent. We will accomplish this using the joint Laplace functional of the two point processes. Consider two arbitrary Poisson point processes, $\Phi$ (homogenous with intensity $\Lambda$) and $\Phi'$. Assume that the motion that transforms $\Phi$ to $\Phi'$ can be represented by a probability kernel $p$, where $p(x, B)$ is the probability that a point located at $x$ in $\Phi$ is located in region $B$ once displacements have occurred to obtain $\Phi'$. We can express the joint Laplace functional of $\Phi$ and $\Phi'$ as
	\begin{align*}
	\mathcal{L}_{\Phi, \Phi'}(f, g)&=\mathbb{E}\left(e^{-\sum_i f(X_i) - \sum_i g(Y_i)}\right)\\
	&=\mathbb{E}\Bigg(  \left(\exp\left[\sum_j{\log{e^{-f(x_j)}}}\right]\right) \left(\exp \left[\sum_j \log \left(\int_{\mathbb{R}^2} e^{-g(y_j)}p(x_j, dy_j)\right) \right]\right)   \Bigg)\\
	&= \mathcal{L}_\Phi(h), \text{ where $h(x) = - \log \left[\int_{\mathbb{R}^2} e^{-f(x) -  g(y)} p(x, dy)\right]$}\\
	&= \exp \left[ - \int_{\mathbb{R}^2} \left(\left(1 - \int_{\mathbb{R}^2} e^{-f(x)-g(y)}\right) p(x, dy)\right) \Lambda(dx)\right]\\
	&=\mathcal{L}_{\Phi'}(g) \text{ exp} \Bigg[-\int_{\mathbb{R}^2}\left(1-e^{-f(x)}\right)\left(\int_{\mathbb{R}^2} e^{-g(y)} p(x, dy)) \right) \Lambda (dx)\Bigg].
	\end{align*}
	Setting $\Phi = \Phi_t$, $\Phi' = \Phi_{t+n}$ and considering the random direction mobility model, we get:
	\begin{align*}
	\mathcal{L}_{\Phi, \Phi'}(f, g)&=\mathcal{L}_{\Phi'}(g) \exp \Bigg[-\int_{\mathbb{R}^2}\left(1-e^{-f(x)}\right)\left(\int_{\partial B(0, v)} e^{-g(x + ns)} G(ds) \right) \Lambda dx\Bigg].
	\end{align*}
	As long as $g(.)$ is such that $g(x) \rightarrow 0$ in all directions as $x \rightarrow \infty$, we have that $\left(\int_{\mathbb{R}^2} e^{-g(x + ns)} G(ds) \right) \rightarrow 1$ as $n \rightarrow \infty$. This follows directly from the Dominated Convergence Theorem, where we will upper-bound the exponential by $1$. This in turn implies that:
	\begin{align*}
	\lim_{n\rightarrow \infty}\mathcal{L}_{\Phi, \Phi'}(f, g) = \mathcal{L}_\Phi(f) \mathcal{L}_{\Phi'}(g).
	\end{align*}
\end{proof}
\begin{corollary}
	\label{interference_mixing}
	The interference shot-noise process, $I_0(t)$, is strongly mixing. 
\end{corollary}
\begin{proof}
Using the notation of Lemma \ref{mixing_condition_compositions}, let $X^{(1)} = \Phi_t$,  and $X^{(2)}, X^{(3)}, ...$ be the fading processes. For any $t$, $X^{(1)}_t = \Phi_t$ is a random measure that takes values in $\mathbb{M}(\mathbb{R}^2)$, the space of locally finite measures on $\mathbb{R}^2$ (see Chapter 1 of \cite{BBK} for more on the formalism of point processes) and $X_t^{(k)}, k>1$ are random variables that take values in $\mathbb{R}$. Further, let $h_t= h \text{ } \forall t$, where $h: \mathbb{M}(\mathbb{R}^2) \times \mathbb{R} \times \mathbb{R} \times ... \rightarrow \mathbb{R}$ and $h(X^{(1)}_t, X^{(2)}_t, X^{(3)}_t,...) = I_0(t) =  \sum_{x_j \in \Phi_t}l(|x_j|)F_j^0(t)$. We argue that $h$ is Borel-measurable by arriving at $h$ via the following compositions: first, an enumeration of the points of $\Phi_t$ in increasing order of their Euclidean distance from the origin - such an enumeration is measurable (see Chapter 1.6.2 in \cite{BBK}) and second, composition with the $L_2$ norm function $|\text{ . }|$ and path-gain function $l(.)$. Each $l(|x_j|)$ so obtained is multiplied by $F_j^0$ and then summed over all $j$, preserving measurability. We can hence apply Lemma \ref{mixing_condition_compositions} to get that $\alpha(X, m) \leq \sum_{n=1}^\infty \alpha(X^{(n)}, m).$  \\From our assumptions on fading, we have that $\alpha(X^{(n)}, m)= 0$ for $m$ larger than the corresponding coherence time and for $n \geq 2$. Theorem \ref{pp_strongly_mixing} implies that $\alpha(X^{(1)}, m)\rightarrow 0$ as $m\rightarrow \infty$. \\
Hence, $\alpha(X, m) \rightarrow 0$ as $m \rightarrow \infty$.
\end{proof}
If the interferer nodes employ power control or contention protocols such that each of their transmission powers or each of their decisions of whether or not to transmit at a given time are measurable functions of only $\Phi_t$ and the fading processes' values at time $t$, a proof similar to the one above can be used to show that the interference process is still strongly mixing.\\
If the receiver is not fixed at the origin but is moving as well, we can still prove a version of the mixing results above. To do so, we place ourselves in the frame of reference of the receiver. The motion of interferers are now no longer independent, but are still conditionally independent if we condition on the motion of the receiver. A proof similar to the one above then can be used, with additional steps for conditioning and unconditioning using the tower rule of expectation. We only require that the motion of the receiver is such that the velocities of interferers in the frame of reference of the receiver are non-zero a.s.\\  
Since strong mixing implies ergodicity, we then have that $s(t)$ and $V(n) = S(n\delta, (n+1)\delta)$ are ergodic processes. We have now established that with non-zero mobility, the driving sequences of the queue, $A(n)$ and $V(n)$ are both ergodic and stationary. Since we also assumed that they are independent of each other, they are jointly stationary and ergodic. These are sufficient conditions to apply the theory of Loynes (\cite{loynes1962stability}) to our queue, and to obtain the main result of this section.
\begin{theorem}
	\label{queue_stability}
	When $v>0$, our queue is stable (workload has a finite limiting distribution) if and only if $\lambda = \mathbb{E}[A(n)] < \delta \mathbb{E}[s(0)]$.
\end{theorem}
Note here that $\mathbb{E}[s(0)]$ is a spatial average over the PPP. As an example, if $s(t) = \log_2[1+\text{SINR}_0(t)]$ and assuming that $\gamma=0$, fading is Rayleigh and $l(r) = (Ar)^{-\beta}, A>0, \beta>2$, the below equation follows from the results presented in Section 16.2.3 of \cite{baccelli2009stochastic2}:
\begin{align*}
\mathbb{E}[s(0)] = \int_0^\infty \exp(-2\pi^2\Lambda R^2 v^{\frac{2}{\beta}}\beta^{-1}(\sin(2\pi/\beta))^{-1})/(v+1))dv.
\end{align*}
\textbf{Discussion: }While the result of Theorem \ref{queue_stability} may seem obvious, we would like to emphasize its significance. What we have shown is that introducing any non-zero mobility, however small, into the system causes the non-ergodic and rather unchanging nature of the queue to vanish - leaving us with a queue that can be stabilized and cope with \textit{``all" Poisson interferer configurations}. Indeed, under the assumption of non-zero mobility, it follows from mixing that the queue will eventually see all possible Poisson interferer configurations. Motion in the network has in a sense unified the spatial averages that are traditionally associated with stochastic geometry and the temporal averages that are traditionally associated with queueing theory. Note the contrast to the infinite or zero mobility models often considered in the analysis of mean local delay - under infinite mobility, the point process decorrelates in a single time slot. In ours, the point process decorrelates after a sufficiently long amount of time  - which, as we show, is still sufficient for a well-behaved queue. In the zero mobility case, it never decorrelates.
\subsection{General Mobility Models and Point Processes}
For the sake of generality, we will briefly discuss how to prove mixing results such as those above for the larger class of mobility models that we described in Section \ref{sec:model}. The key step is to prove an equivalent of Theorem \ref{pp_strongly_mixing} for the specific mobility model under consideration. More specifically, given that the displacement experienced by a point over $n$ time steps follows a law denoted by $\Xi_n$, one needs to prove that the point processes $\Phi_t$ and $\Phi_{t+n}$ are asymptotically independent in the limit $n \rightarrow \infty$. For example, in the case of Brownian motion, $\Xi_n$ consists of two independent Gaussian displacements in the $x-$ and $y-$directions, each with zero mean and variance $n\sigma^2$. The mixing proof for this case again uses the Dominated Convergence Theorem to establish asymptotic independence as $n\rightarrow \infty$.  \\
\textcolor{black}{Finally, we would like to note that Theorem \ref{queue_stability} can hold for a larger class of point processes than the Poisson point process that we consider in this paper. Indeed, it will hold for any initial point process $\Phi_0$ and any mobility model such that $\Phi_t$ is ergodic. We make the Poisson assumption and the assumptions on the mobility models because they allow us to analytically prove strong mixing and ergodicity, but the results need not be restricted to these cases alone.}
	\section{Mobility Alleviates Queue Workloads}
	\label{sec:mobility_orders_workloads}
	Now that we are assured of the queue's stability under appropriate conditions on the input data rate, we will investigate the effect of increasing degrees of mobility on the statistics of the queue workload. To begin, we define useful partial orderings.
	\begin{definition}
		\label{cx_def}
		\textbf{Convex order }($\leq_{cx}$): Consider two random variables $X$ and $Y$. $X$ is said to be smaller than $Y$ in the convex order, denoted by $X \leq_{cx} Y$, if $\mathbb{E}[f(X)] \leq \mathbb{E}[f(Y)]$ for all convex functions $f$, provided expectations exist.
	\end{definition}
	\begin{definition}
		\label{icx_def}
		\textbf{Increasing convex order }($\leq_{icx}$): Consider two random variables $X$ and $Y$. $X$ is said to be smaller than $Y$ in the increasing convex order, denoted by $X \leq_{icx} Y$, if $\mathbb{E}[f(X)] \leq \mathbb{E}[f(Y)]$ for all increasing convex functions $f$, provided expectations exist.
	\end{definition}
	For a further treatment on convex orderings, see \cite{gupta2010convex}. A useful equivalent definition for these orderings is presented in the lemma below (taken from \cite{baccelli2013elements}, Chapter 4).
	\begin{lemma}
		\label{eq_def}
		Given two random variables $X$ and $Y$, $X \leq_{cx} (\text{resp. } \leq_{icx}) \text{ } Y$ if and only if there exist two random variables $A$ and $B$, identically distributed as $X$ and $Y$ respectively and defined on a common probability space $(\Omega, \mathcal{F}, \mathbb{P})$, such that $A = (\text{resp. }\leq)\text{ } \mathbb{E}[B|\mathcal{G}]$ a.s., for some sub $\sigma$-field $\mathcal{G}$  of $\mathcal{F}$.
	\end{lemma}
	Now, we will consider 3 environments (an environment is a queue and its moving interferers). These environments differ only in the velocities of the interferers and are coupled otherwise - they have the same arrival process, $A^{(1)}(n) = A^{(2)}(n) = A^{(3)}(n) = A(n)$. Let the velocities be $v_1$, $v_2$ and $v_3$ respectively, with $v_2 = mv_1$ and $v_3 = nv_1$ and $n>m$ ($n, m \in \mathbb{N}$, $n>1$, $m>1$), so that $v_3 > v_2 > v_1$. Let $s^{(i)}(t)$ be the instantaneous service rate that the queue can provide in Environment $i$. Let the cumulative service that the queue can provide in the time interval $(t_1, t_2]$ for Environments 1, 2 and 3 be $S^{(1)}(t_1, t_2)$, $S^{(2)}(t_1, t_2)$ and $S^{(3)}(t_1, t_2)$ respectively. Let the workloads of the queue at time slot $n$ be $W^{(1)}_n$, $W^{(2)}_n$ and $W^{(3)}_n$ respectively.\\
	The key observation now is that an increase in the velocity of interferers is equivalent to an acceleration of time while keeping velocity fixed. Consider a simple example - the distance an interferer moving with velocity $2v$ will cover in time $\Delta t$ is the same as the distance an interferer moving with velocity $v$ will cover in time $2\Delta t$. In the context of the environments described above, scaling velocities by a factor of $m \in \mathbb{N}$ is equivalent to accelerating time by a factor of $m$. Hence, we have
	\begin{align*}
	&S^{(2)}(t_1, t_2) = \int_{t_1}^{t_2} s^{(1)}(mt) dt = \frac{1}{m}\int_{mt_1}^{mt_2} s^{(1)}(t) dt\\ 
	&= \frac{1}{m}\sum_{i=1}^{m} \int_{mt_1 + (i-1)(t_2-t_1)}^{mt_1 + (i)(t_2-t_1)} s^{(1)}(t) dt\\
	&= \frac{1}{m}\sum_{i=1}^{m} S^{(1)}(mt_1 + (i-1)(t_2-t_1), mt_1+i(t_2-t_1))\\
	&= \frac{1}{m}\sum_{i=1}^m Y^{(2)}_i,
	\end{align*}
	where the variables $Y^{(2)}_i = S^{(1)}(mt_1 + (i-1)(t_2-t_1), mt_1+i(t_2-t_1))$ are identically distributed (since $s(t)$ is stationary) but not independent. Similarly, we have
	\begin{align*}
	S^{(3)}(t_1, t_2) = \frac{1}{n} \sum_{j=1}^n Y^{(3)}_j,
	\end{align*}
	for $Y^{(3)}_j=S^{(1)}(nt_1 + (i-1)(t_2-t_1), nt_1+(i)(t_2-t_1))$, where \{$Y^{(3)}_j$\} is distributed identically to the variables \{$Y^{(2)}_i$\}. Note also that $S^{(1)}(t_1, t_2) = Y^{(1)}_j$ is identically distributed to \{$Y^{(2)}_j$\} and \{$Y^{(3)}_j$\}. Henceforth, we will drop the superscript for the variables $Y_i$, and write
	\begin{align}
	\label{sample_mean_interpretation}
	S^{(1)}(t_1, t_2) &=Y_i, 
	S^{(2)}(t_1, t_2) = \frac{1}{m} \sum_{i=1}^m Y_i, \nonumber \\
	S^{(3)}(t_1, t_2) &= \frac{1}{n} \sum_{i=1}^n Y_i. 
	\end{align}
	Now, let $M^{(i)}_d \in \mathbb{R}^d$ be a vector whose $j^{th}$ element is $S^{(i)}(j\delta t, (j+1)\delta t)$, for $d>j\geq0$. Then, we have the following lemma.
	\begin{lemma}
		\label{cvx_ordering_service}
		$M^{(3)}_d \leq_{cx} M^{(2)}_d \leq_{cx} M^{(1)}_d$ for all $d>0$.
	\end{lemma}
	\begin{pf}
		From (\ref{sample_mean_interpretation}), we see that we can write  	
		$M^{(1)}_d =Z_i, M^{(2)}_d = \frac{1}{m} \sum_{i=1}^m Z_i, \text{and } M^{(3)}_d = \frac{1}{n} \sum_{i=1}^n Z_i, $ where $Z_i$ is an $d-$dimensional vector, with each element of the vector distributed identically to $Y_0$ (but not necessarily independent of the other elements in the vector). Further define ${X}_n = \sum_{i=1}^n Z_i$ and $\overline{X}_n = \frac{X_n}{n} \in \mathbb{R}^d$. We will now show that $\overline{X}_{n-1} \geq_{cx} \overline{X}_{n}$. The result will then follow by the transitivity of the convex ordering.\\ First, note that $\mathbb{E}[Z_i|X_n]$ is a function of $X_n$, and is independent of $i$ due to the stationarity of $\{Z_i\}$. Let $\mathbb{E}[Z_i|X_n] = \Gamma(X_n)$. Then,
		\begin{align*}
		\Gamma(X_n) &= \frac{1}{n}\sum_{i=1}^n \mathbb{E}[Z_i|X_n] = \mathbb{E}\left[\frac{X_n}{n}|X_n\right]= \overline{X}_n.
		\end{align*}
		Since $\mathbb{E}[Z_i|X_n] = \overline{X}_n$, we have that 
		\begin{align*}
		\mathbb{E}[Z_1 + Z_2 + ... +Z_{n-1}| X_n] &= (n-1)\overline{X}_n,
		\end{align*}
		{which implies } $\mathbb{E}[\overline{X}_{n-1}|\overline{X}_n]= \overline{X}_n$.
		From Lemma \ref{eq_def}, we have that $\overline{X}_{n-1} \geq_{cx} \overline{X}_{n}$.
	\end{pf}
	The intuition behind Lemma \ref{cvx_ordering_service} is straightforward - sample means calculated using a larger number of samples will possess lower variability, resulting in a convex ordering. Hence, an accelerated service process will result in the queue "seeing" a larger number of samples of instantaneous service rate, which in turn will achieve an averaging effect that results in more reliable service being provided to the queue. It follows that this increasing reliability of service will reflect in the performance of the queue, a qualitative result that we will present in Theorem \ref{queue_order_thm}.\\
	Now, consider two sequences of variables $\{\widetilde{y}_i\}$ and $\{{y_i}\}$, $i\geq0$, whose evolution is governed by the same stochastic recurrence function $h$, and two sets of driving sequences $\{\widetilde{\beta}_n\}$ \& $\{\beta_n\}$ and initial conditions $\widetilde{y}_0$ \& $y_0$ respectively:
	\begin{align*}
	y_{n+1} &= h(y_n, \beta_n)\\
	\widetilde{y}_{n+1} &= h(\widetilde{y}_n, \widetilde{\beta}_n).
	\end{align*}
	These sequences have the same dynamics, characterized by $h$, and differ only in their driving sequences and initial conditions. Then, we have the following lemma:
	\begin{lemma}
		\label{cvx_queue}
		(\cite{baccelli2013elements}, Property 4.2.5): Assume that the driving sequences and initial conditions are integrable, and that the function $y \rightarrow h(y, \beta)$ is non-decreasing and that the function $(y, \beta) \rightarrow h(y, \beta)$ is convex. Then, $(y_0, \beta_0, \beta_1, ...) \leq_{cx} (\widetilde{y}_0, \widetilde{\beta}_0, \widetilde{\beta}_1, ...)$ implies $(y_0, y_1, y_2, ...) \leq_{icx} (\widetilde{y}_0, \widetilde{y}_1, \widetilde{y}_2, ...)$.
	\end{lemma}
	We now present the main theorem of this section.
	\begin{theorem}
		\label{queue_order_thm}
		For queues that start with initial workload $W^i_0=0$ and evolve according to the Lindley recursion (\ref{lindley_recursion}), and for environments as defined previously, 
		\begin{align}
		\label{queue_workload_icx}
		W^{(3)}_n \leq_{icx} W^{(2)}_n \leq_{icx} W^{(1)}_n, \forall{}n.
		\end{align}
		Further, the steady state workloads that the queues converge to in the limit are also similarly ordered:
		\begin{align}
		W^{(3)}_\infty \leq_{icx} W^{(2)}_\infty \leq_{icx} W^{(1)}_\infty.
		\end{align}
	\end{theorem}
	\begin{pf}
		The convex ordering is preserved by the operation of multiplication by $-1$ and by the addition of an independent random variable. Let $A_n \in \mathbb{R}^n$ be a vector whose $j^{th}$ element is $A(j)$, for $n > j \geq 0$. From Lemma \ref{cvx_ordering_service},
		\begin{align*}
		A(n)-M^{(3)}_n \leq_{cx} A(n)-M^{(2)}_n \leq_{cx}	A(n)-M^{(1)}_n.
		\end{align*}
		Now, we define $h(y, \beta) = (y + \beta)^+$, and set $y_n$ to be queue workload $W_n$, initial conditions $y_0 = \widetilde{y}_0 = 0$ and $(\beta_0, \beta_1, ..., \beta_{n-1}) = A_n - M_n$. Then the proof of the first part of the theorem follows from Lemma \ref{cvx_queue}. Since we are operating in the regime where the queue is stable, we know that the queue workload will converge to a stationary steady state variable $W^{(i)}_\infty$. The second part of the theorem then follows from a direct application of the Monotone Convergence Theorem (see \cite{baccelli2013elements}, Section 4.2.6 for more details).
	\end{pf}
	We emphasize again here that this result holds for queues that are driven by \textit{correlated} service processes, and is a significantly more general result than results of the Pollaczek-Khinchine type, which hold only for independent service processes.\\
	A first consequence of this theorem is the following corollary.
	\begin{corollary}
		\label{mean_workload_ordering}
		$\mathbb{E}[W^{(3)}_\infty] \leq\mathbb{E}[W^{(2)}_\infty] \leq\mathbb{E}[W^{(1)}_\infty]$.
	\end{corollary}
	It also follows that the mean delay that packets see are similarly ordered. Let us assume that a packet is successfully transmitted once all the data associated with it is transmitted, and the delay a packet sees is the length of the time interval from its arrival to its succesful transmission. Recalling that we assume that each packet is of unit size, at time slot $n$, the (fractional) number of packets in the queue will be $N_\infty = W_\infty$. ${N}_\infty$ is hence ordered in expectation similar to $W_\infty$. It follows from Little's Law that mean delay, $\mathbb{E}[D] = \mathbb{E}[N_\infty]/\mathbb{E}[A(n)]$ is also ordered in the same way, provided that it exists. 
	\begin{corollary}
		\label{mean_delay_ordering}
		$\mathbb{E}[D^{(3)}] \leq\mathbb{E}[D^{(2)}] \leq\mathbb{E}[D^{(1)}]$.
	\end{corollary}
	\subsection{Discussion}
	These results state that queues that see the same arrival process see decreasing workloads and delays as they are placed amongst interferers that move with increasing mobility. Hence, the performance of queues in such a wireless network improves as the degree of mobility in the network increases. This performance improvement can be intuitively thought of as a consequence of mobility introducing diversity into the network. The diversity that we refer to is with respect to interferer configurations - faster moving interferers allow a larger number of configurations to be seen at a higher rate. In other words, the network mixes faster. An intuitive explanation that is a dual to the faster averaging/mixing argument is the following scenario. Assume the queue is surrounded by a "bad" point process of interferers - such that a large number of interferers are located close to the origin. The interference seen at the origin will hence be large. A slowly moving environment will cause the interference at the origin to be large for a long period of time, which will result in a build-up in the workload of the queue. On the flip side, interferer configurations that are "good" will also persist for long periods of time, allowing the queue to drain effectively during those periods. A fast-moving environment, on the other hand, will result in shorter build-ups and drainings of the queue. Fig. \ref{slow_and_fast} shows the difference in the lengths of these cycles and hence in the average queue workload for environments with different degrees of mobility.
	\begin{figure}[h]
		\begin{center}
		\begin{subfigure}[h]{0.45\columnwidth}
			\includegraphics[width=\linewidth]{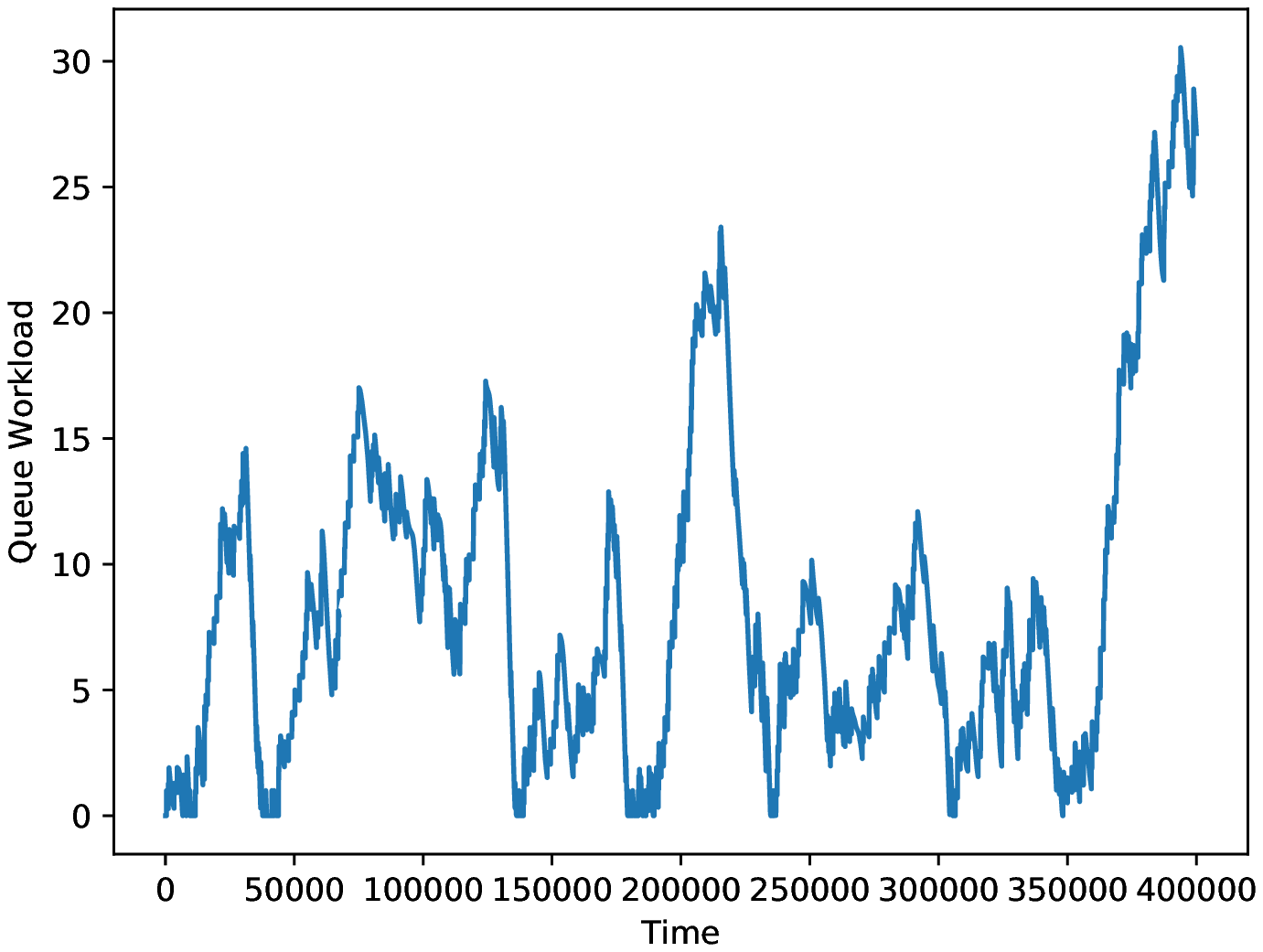}
			\caption{Slow environment\\ ($v=1$)}
			\label{fig:gull}
		\end{subfigure}
		\begin{subfigure}[h]{0.45\columnwidth}
			\includegraphics[width=\linewidth]{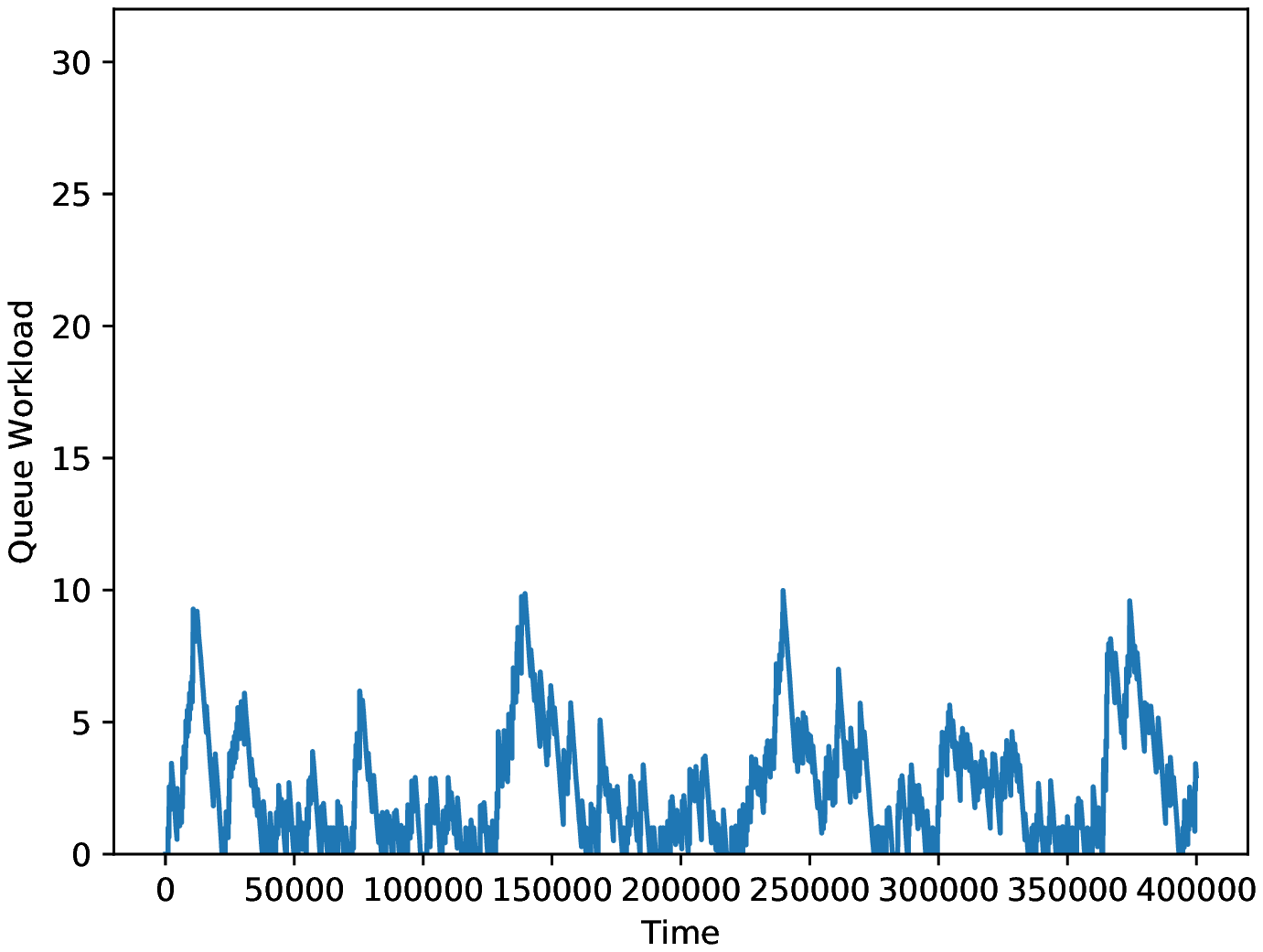}
			\caption{Fast environment\\ ($v=1000$)}
			\label{fig:gull2}
		\end{subfigure}
	\end{center}
\caption{Comparison of queue build-up-and-drain cycles for different degrees of motion in the random direction model}\label{slow_and_fast}
	\end{figure}
	
	 From an expectation point-of-view, it seems clear that longer build-up-and-drain cycles will result in larger mean workloads and mean delays, which is what we see in Corollary \ref{mean_workload_ordering} and \ref{mean_delay_ordering}. Interestingly, the more general result of Theorem \ref{queue_order_thm} also holds - increasing convex functions of workloads are also similarly ordered. \\
	We would like to draw attention to the generality of results presented in this section. First, in order to compare the workloads of two environments, the velocity of one environment must be a rational multiple of the other. Since the rationals are dense in the reals, we do not lose any generality here. In terms of modelling, note that we make no assumptions on the service process $s(t)$ save that it is integrable and a measurable function of $I_0(t)$. We only need a stationarity and mixing assumption for the fading processes. These results hold for all the mobility models described in Section \ref{sec:model}. Indeed, these results will hold for \textit{any} general mobility model that admits the interpretation of an increase in velocity as an acceleration in time\textcolor{black}{, and for \textit{any} initial point process $\Phi_0$ (not necessarily Poisson)}. The case of the receiver at the origin being mobile as opposed to static is subsumed under this set of mobility models - the relative velocities of interferers with respect to the tagged receiver will no longer be independent of each other, but will still be amenable to the acceleration-of-time framework used in this section. Models that incorporate power control/contention protocols can also be analyzed using the acceleration-of-time framework as long as the transmission powers/transmission decisions are functions of $\Phi_t$.\\
	We briefly discuss the limiting case of infinite velocity. Under this mobility assumption, the point process of interferers will be resampled at every time instant. This will decorrelate the service process - the variables $Y_i$ will now be i.i.d. and the service process seen in a time interval $(t_1, t_2)$ will be a constant - the infinite averaging will result in the expected service rate being seen. This results in a G/D/1 queue. Making further assumptions on the arrival process and network model could lead to analytical expressions for queue statistics in this infinite mobility case, but we will not pursue this here.

	\section{Correlation structures in mobile networks}
	\label{sec:correlation}
	In this section, we formalize the correlations that we discussed while presenting the intuition of build-up-and-drain cycles for the results in the previous sections. First, we consider SINR level-crossing events - $L_t$ is defined as the event that at time $t$, $\text{SINR}_0(t) > T$ for some fixed threshold $T$. We will show that if we see a level-crossing event at time $t$, there is an increased likelihood of seeing another level-crossing event in some neighbourhood of that time. Assume Rayleigh fading between the transmitter and receiver of interest (modelled by exponential random variables of mean $1/\mu$), and let all other fading variables be generally distributed but with unit mean. We then have the following theorem.
	\begin{theorem}
		\label{positive_correlation_crossing}
		$\mathbb{P}(L_{t+1} | L_t) > \mathbb{P}(L_{t+1}) = \mathbb{P}(L_t).$
	\end{theorem}
	\begin{proof}
		Presented in the appendix.
	\end{proof}
	Theorem \ref{positive_correlation_crossing} formalizes the intuition we presented in Section \ref{sec:mobility_orders_workloads}. It implies that if we observe a high (resp. low) SINR, we are likely to continue to observe a high (resp. low) SINR for a while. \textcolor{black}{The notion of build-up-and-drain cycles of the queue is a direct consequence of this effect - the larger the positive correlations in the system, the longer the build-up-and-drain cycles and the larger the mean queue length and delay. This positive temporal correlation of level-crossing events will decrease as the correlation between the point processes $\Phi_t$ and $\Phi_{t+1}$ decreases - or in other words, as mobility increases - resulting in the stochastic orderings we established in Section \ref{sec:mobility_orders_workloads}.}\\
	Finally, for completeness, we will derive the correlation coefficient between the interference shot-noise observed at times $t$ and $t+1$, $\gamma(t, t+1)$, defined as:
	\begin{equation*}
	\gamma(t, t+1) = \frac{\mathbb{E}[I_0(t)I_0(t+1)] - \mathbb{E}[I_0(t)]^2}{\mathbb{E}[I_0(t)^2] - \mathbb{E}[I_0(t)]^2}.
	\end{equation*}
We will use $\mathcal{L}_{t, t+1}(.,.)$ to denote the joint Laplace transform of $I(t)$ and $I(t+1)$.
\begin{lemma}
	\label{positive_correlation_coefficient}
		\begin{equation}
	\label{corr_coeff}
	\gamma(t, t+1) = \frac{\int l(x) \int l(y) p(x, dy) dx}{\mathbb{E}[h^2] \int l(x)^2 dx}.
	\end{equation}
\end{lemma}
\begin{proof}
We use Campbell's formula to calculate the first and second moments of interference, and then use the fact that $\mathbb{E}[I_0(t)I_0(t+1)] = \frac{ \partial^2 \mathcal{L}_{t, t+1}(s_1, s_2)}{\partial s_1 \partial s_2}\big\rvert_{s_1=0, s_2=0}$.
\end{proof}
	We can see that the correlation is highest when the probability kernel that defines mobility is such that $p(x, B)=\delta_x(B)$, and is lowest in the limit of infinite mobility, \textcolor{black}{and decreases with increasing velocity}.
	
\section{The Interacting Queues Setting}
\label{sec:interacting}
As we mentioned before, our goal is to study the effect of mobility on wireless networks using queueing metrics as performance measures. The model we have discussed thus far - that of a single queue whose evolution is coupled with the interference it experiences from a network of interferers with finite mobility - is a novel one, and is the necessary first step in achieving our goal. The holy grail, however, is a more general model where interferers themselves are also queues. Such a model immediately possesses time-correlations that are far more complex than the ones we deal with in previous sections. Nonetheless, we believe that the basic insights we present in this work will apply to that model as well\textcolor{black}{, and our intuition is supported by simulation results that we describe later in this section}.\\Consider the same model as presented in Section \ref{sec:model}, except that every interferer now has a queue of data to transmit. For simplicity, let us assume that every interferer is the transmitter in a moving transmitter-receiver pair, with a queue of data to transmit to its dedicated receiver. We'll assume that the receiver at the origin and its associated transmitter are not moving (This is a simplifying assumption, but as we have mentioned in previous sections, it can be lifted with a little work). Assume that data arrives at the same average rate to all queues \textcolor{black}{(we will call this a homogeneous system, as opposed to a heterogeneous system where average arrival rates are different across queues)}. The SINR at the origin is now 
	\begin{align}
\label{sinr_origin_interacting}
\text{SINR}_0(t) &= \frac{S_0(t)}{I_0(t) + \sigma^2(t)}
= \frac{l(R)F_0^0(t)}{\sum_{x_j \in \Phi_t}\mathcal{X}_j(t) l(|x_j|)F_j^0(t) + \gamma}.
\end{align}
This is similar to (\ref{sinr_origin}), but with the extra set of time-evolving processes $\{\mathcal{X}_j(t)\}$. $\mathcal{X}_j(t)$ is 1 if the queue associated with interferer $j$ is non-empty at time $t$, and is 0 otherwise. Similar expressions can be written for the SINRs seen at the receivers associated with interferers. Clearly, the evolution of all the queues are coupled with each other - first because the interference each queue sees is a function of the others' locations, but additionally because the interference each queue sees is a function of whether all the other queues are empty or not. Again, let us consider the two questions we analyzed previously. \\It is easy to see that in the case without motion, i.e., when $\Phi_t = \Phi_0 \forall t$, we cannot provide stability guarantees for any arrival rate for similar reasons as those we state in Section \ref{sec:mobility_guarantees_stability}. Given enough time \textcolor{black}{or starting with sufficiently large initial conditions}, closely clustered queues will start to build up in length and interfere excessively with each other. \textcolor{black}{An instability result similar to Proposition \ref{prop:instability} is hence straightforward to state. }Any non-zero mobility will again intuitively alleviate these bad configurations, but a formal proof of strong mixing similar to Corollary \ref{interference_mixing} is elusive, since now the variables $\{\mathcal{X}_j(t)\}$ for a particular value of $t$ are correlated by the processes $\Phi_s, F_1^0(s), F_2^0(s), ...$ $\forall s \leq t$. \textcolor{black}{Nonetheless, thorough simulation studies of such a homogeneous system indicate that as long as there is non-zero mobility ($v>0$) of interferers, the necessary and sufficient condition for stability is again the same as what is described in Theorem \ref{queue_stability}. This is in line with existing results on homogenous interacting queue systems (e.g., \cite{bonald2004wireless}). Stability conditions for heterogeneous systems are significantly more difficult to analyze, and very little is known for any more than two interacting queues (see \cite{fayolle1979two} and \cite{cohen2000boundary} for the two-queue problem, and \cite{cohen1984functional} for a discussion on larger systems).} \\To discuss the question of how statistics of all queues will change with different degrees of mobility, we recall the intuition illustrated in Fig. \ref{slow_and_fast} and formalized in Section \ref{sec:correlation}. Increasing mobility in the network will decrease correlations across time instants. This, we believe, will again have the effect of improving queue statistics as mobility increases - but we have not been able to prove this. \textcolor{black}{Simulation studies of the interacting queue system support this claim, since we see that the first moment of queue length/workload decreases with increasing velocity of interferers (Fig. \ref{fig:interacting_queue_workload}, description of simulations in Section \ref{sssec:interacting_sim}).} In the extreme limit of infinite mobility, queues decouple from correlations between each other and across time, and evolve independently. \\
\textcolor{black}{To conclude, we note that interacting queue systems have been studied for over half a century and are often not amenable to exact analysis, as mentioned earlier. The insights we obtain from our model allow us to make educated guesses about an interacting queue version of the system which are supported by simulations.}
	\section{Analytical and Simulation studies}
	\label{sec:simulation}
	\subsection{Average Workload Studies}
	\label{ss:avg_workload}
	The model is that of Section \ref{sec:model} and our simulation setup is as follows \textcolor{black}{(all simulations are carried out using Python)}. We consider a 100 $\times$ 100 square whose edges are wrapped around to avoid edge effects. In order to approximate continuous time, we let our overall system evolve in very small intervals of discrete time, of length $\Delta  = 10^{-3}$ seconds. The queue evolves in time slots of length $\delta $. We set $\delta  = \Delta $ in simulations, which can be interpreted as approximating a continuous time queue. The arrival process $A(n)$ to the queue is a Bernoulli process of rate $\lambda$, where at every time slot a packet arrives with probability $\lambda \delta $, with $\lambda$ chosen such that $\lambda \delta <1$. Interfering nodes are distributed as a PPP, and move according to the mobility models described before. We assume that the service process for the queue is the truncated Shannon rate process described in Section \ref{sec:model}. The reduction in workload that the queue sees at time $t$ in the interval $\Delta $ is hence $\log_2(1+\text{SINR}_0(t))\mathbbm{1}[\text{SINR}_0(t)>T] \Delta $. The queue evolves according to (\ref{lindley_recursion}). From Theorem \ref{queue_stability}, to guarantee stability of the queue, we must ensure that $\mathbb{E}[A(n)] = \lambda < \mathbb{E}[s(0)]$ (since $\Delta  = \delta $). We empirically estimate $\mathbb{E}[s(0)]$ for our system and set $\lambda$ accordingly. For our simulations, we set $\Lambda = 0.1$, noise power $\sigma = 0$,  $R=0.3$, $T=8$ and $l(r) = (1+r)^{-4}$. We consider Rayleigh fades with mean 1.\\To begin, we verify the ordering of mean queue workload that is indicated by Corollary \ref{mean_workload_ordering}. We estimate the mean workload, and 95\% confidence intervals using the batch mean method detailed in Chapter IV.5 of \cite{asmussen2007stochastic}.
		\begin{figure}[t]
	\begin{center}
		{\includegraphics[width=0.7\textwidth]{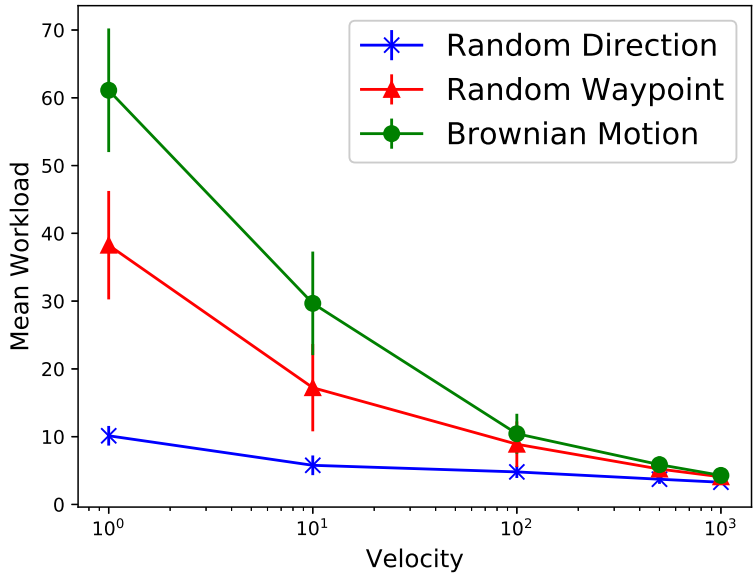}}
		\caption{Comparison of the effects of different mobility models. $\lambda = 1.2$, $\mathbb{E}[s(t)] \approxeq 1.37$, $\Lambda = 0.1$.} 
		\label{fig:mobility_comparison_fig}
	\end{center}
\end{figure}
	Fig. \ref{fig:mobility_comparison_fig} shows the anticipated trend in mean workload, $\mathbb{E}[W]$, for all mobility models. Mean delay, $\mathbb{E}[D] = \mathbb{E}[W]/\lambda$, will show the same trend. We see an order of magnitude difference between the mean workload for $v=1$ and the mean workload for $v=1000$. Next, we compare the effects of different mobility models on mean queue workload. This comparison can be thought of as a measure of how fast each mobility model causes the PPP of interferers to mix - the faster the network mixes, the lower the mean workload will be. Fig. \ref{fig:mobility_comparison_fig} also compares the mobility models presented in Section \ref{sec:model}. We see that workloads associated with Brownian motion (BM) dominate those associated with the random waypoint model (RWP), which in turn dominate those associated with the random direction (RD) model. \textcolor{black}{Here is an intuitive explanation for why we see this ordering among mobility models. Mixing is essentially the decorrelation of the point process from its initial configuration ($\Phi_0$). A point process with RD mobility decorrelates faster than one with RWP mobility. This is because each point undergoing RD motion moves in a fixed direction for all time, away from its initial position. Points undergoing RWP motion change direction with time, and will often end up moving towards their initial positions, \textit{decreasing} displacement from their initial positions and slowing the rate of decorrelation. Brownian motion can be thought of as RWP motion with instantaneous changes in direction, resulting in even slower decorrelation.} Note that the RD model is used in the literature to model vehicular networks (see \cite{choi2018poisson}). In contrast, the BM and RWP models have been used to model more erratic movement, such as that of pedestrians. Our simulation studies suggest that the former motion leads to better performance of wireless queues, due to its more directed nature.
	\subsubsection{Interacting Queue Simulations}
	\label{sssec:interacting_sim}
	\textcolor{black}{We simulate an interacting queue system (as described in Section \ref{sec:interacting}) where all interferers are themselves queues. We choose $\Delta = 10^{-2}$, $\delta = 1$, $\Lambda = 0.01$, noise power $\sigma = 0.1$,  $R=0.3$, $T=8$ and $l(r) = (1+r)^{-4}$. We consider Rayleigh fades with mean 1. The arrival processes for all queues are chosen as i.i.d. Bernoulli processes of rate $\lambda$, where at every time slot of length $\delta=1$ a packet arrives to a queue with probability $\lambda = 0.08$. $\mathbb{E}[s(0)]$ can be computed and has the value $0.0907$. Mean queue length is plotted vs. velocity in Fig. \ref{fig:interacting_queue_workload} with 95\% confidence intervals calculated using the batch mean method.}
	\begin{figure}[t]
		\begin{center}
			\includegraphics[width=0.6\textwidth]{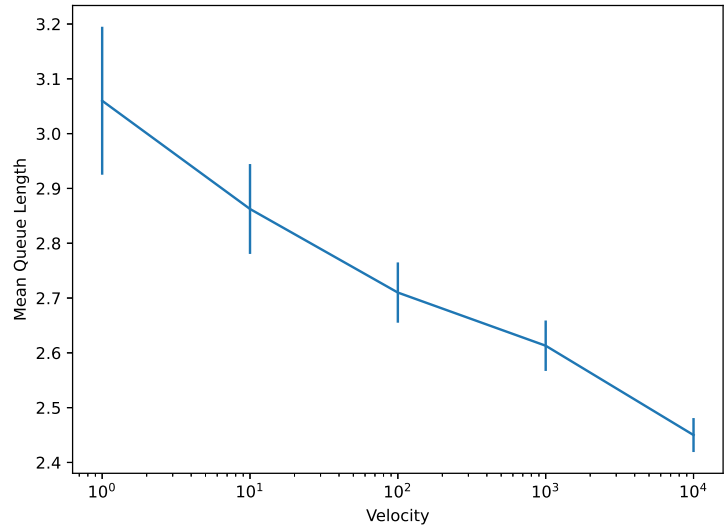}
			\caption{Mean queue length vs. interferer velocity for an interacting queue system.}
			\label{fig:interacting_queue_workload}
		\end{center}
	\end{figure}
		\subsection{Delay analysis}
		\label{ss:delay_analysis}
	\begin{figure}[th]
		\begin{center}
			\begin{subfigure}[h]{0.49\columnwidth}
				\includegraphics[width=\linewidth]{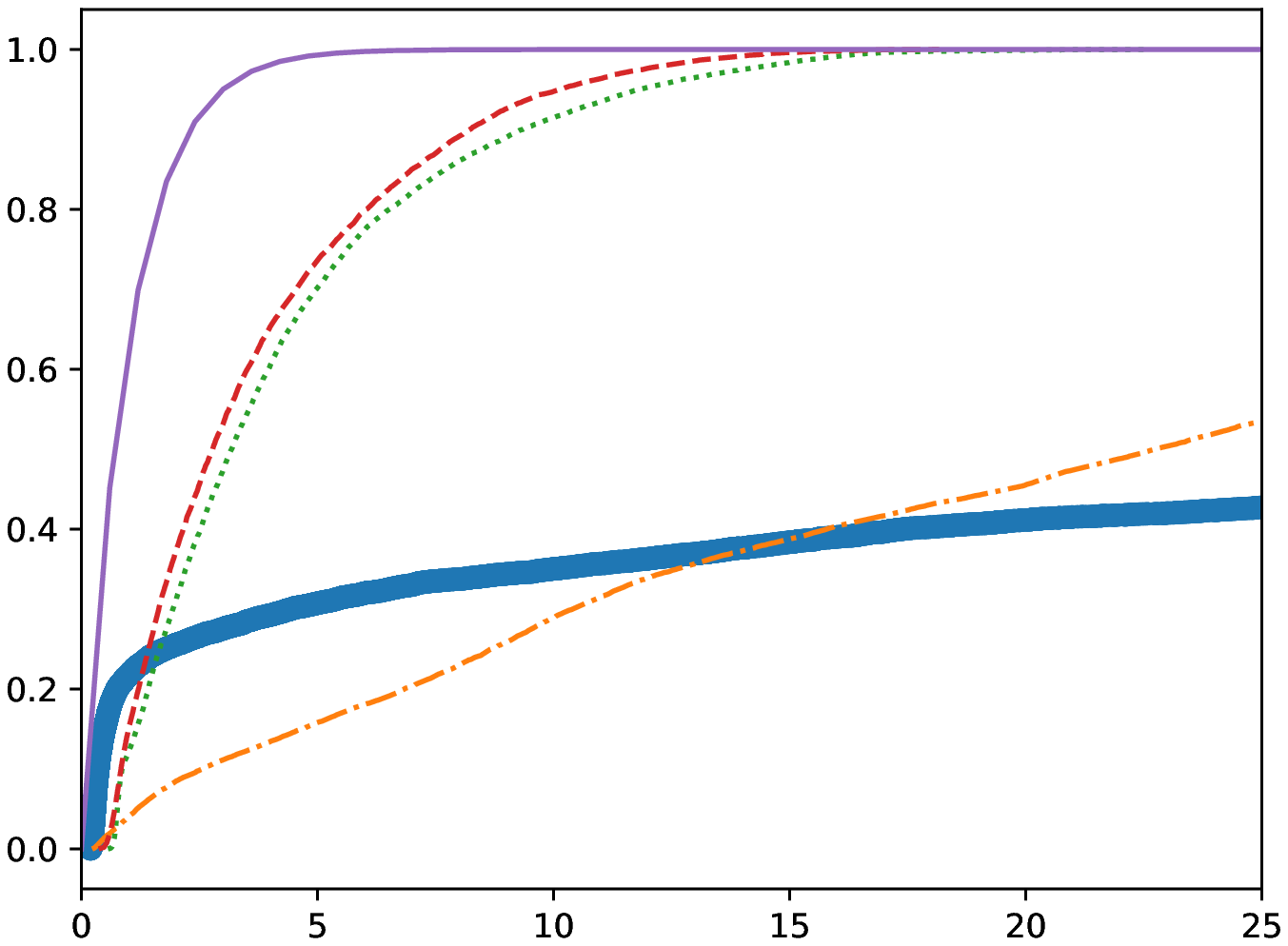}
				\caption{X-axis = [0, 25]}
				\label{fig:near}
			\end{subfigure}
			\begin{subfigure}[h]{0.49\columnwidth}
				\includegraphics[width=\linewidth]{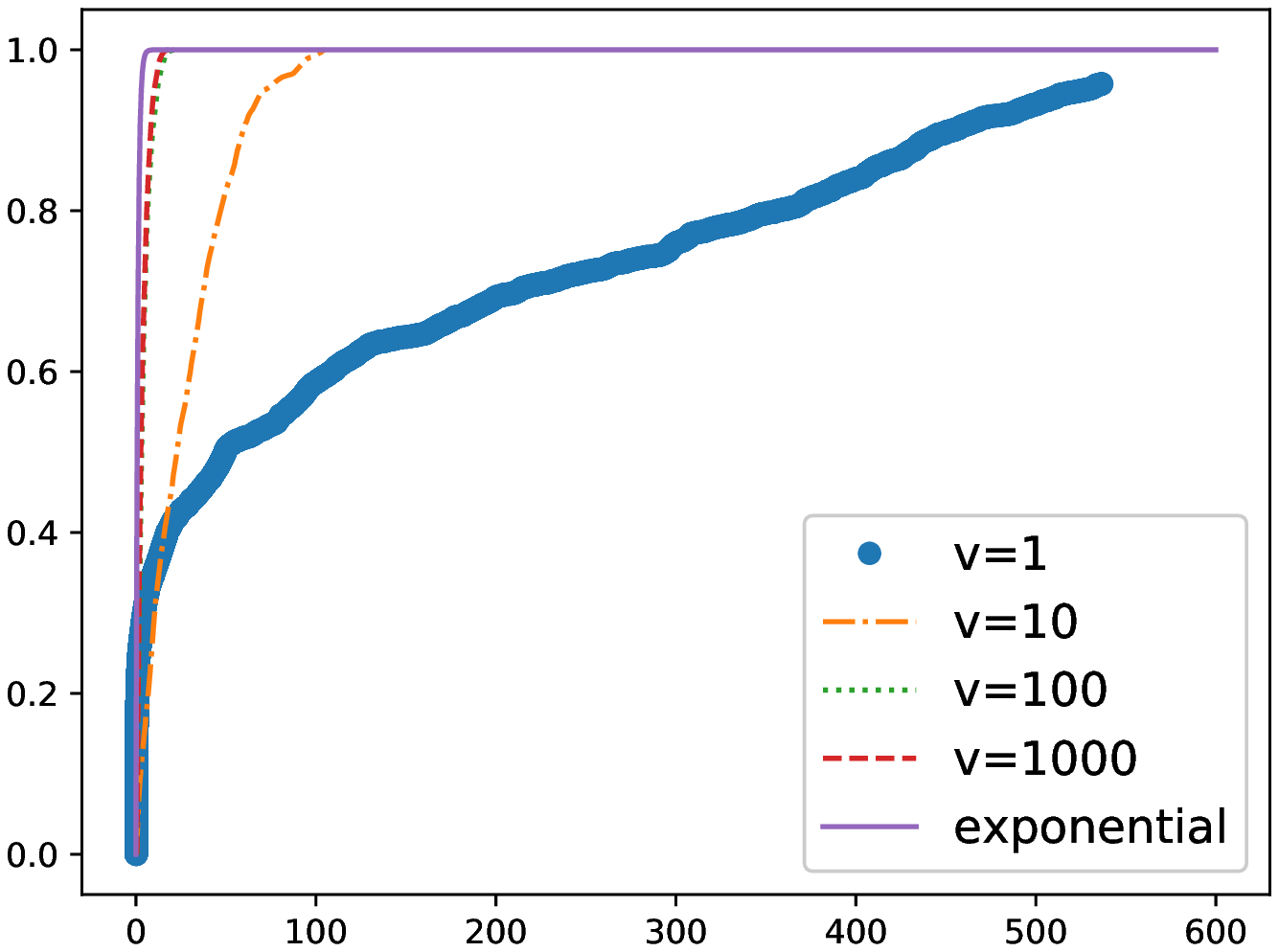}
				\caption{X-axis = [0, 600]}
				\label{fig:far}
			\end{subfigure}
		\end{center}
		\caption{Empirical CDF of latency, observed at two different scales. Legends are the same for both plots.}\label{near_and_far}
	\end{figure}
	We also evaluate and plot the empirical CDF of latency/delay (as defined before Corollary \ref{mean_delay_ordering}) for the random direction model and the values of parameters used in the previous subsection. We observe via comparison of the CDFs that the tail distributions of latency become heavier with decreasing velocities. This observation has engineering implications for communication networks - in the context of our system, for example, this indicates that packets will more often see large latencies with decreasing mobility levels. A system design implication is that higher velocities will be advantageous in scenarios where reliable service with low probability of rare events is required, such as in URLLC systems.
	\subsection{Heavy-traffic approximations for mean queue workload}
	\label{ss:heavy_traffic}
	We would like a computational method that can estimate the mean workload of queues experiencing different interferer mobility levels. To obtain such expressions, we use heavy-traffic approximation results - our analysis is hence a good approximation only when the queue experiences heavy traffic. This approach was first established by Kingman \cite{kingman1962queues} and later generalized by Whitt and others (see \cite{whitt2000overview}, \cite{whitt1974heavy}). These results provide approximations for quantities of the queue as the load factor $\rho$ of the queue tends to $1$, i.e., when the queue experiences heavy traffic. The results use functional central limit theorems to establish stochastic-process limits of heavy-traffic queues, and then use properties of the limits established to obtain approximations for the queueing dynamics. For more details, see \cite{whitt2002stochastic} and references therein.\\Using the results established in \cite{whitt2002stochastic} (Chapter 9.6), we have that
	\begin{align}
	\mathbb{E}[W_\infty] \approx \frac{\mathbb{E}[A(1)]\rho(c_A^2 + c_S^2)}{2(1-\rho)},
	\end{align}
	where 
	\begin{gather}
	\rho = \mathbb{E}[A(1)]/\mathbb{E}[V(1)]\nonumber \\
	c_A^2 = Var[A(1)]/\mathbb{E}[A(1)]^2, \text{ } c_S^2 = \lim_{k \rightarrow \infty} \frac{1}{k} \sum_{j=1}^k (k-j) \frac{Cov(V(1), V(j))}{\mathbb{E}[V(1)]^2}\label{eq:c_s_2}.
	\end{gather}
	The effect of interferer mobility on the evolution of the queue is captured in $c_S^2$ - as the mobility of interferers increases, the correlations between service amounts in different time slots will decrease, and $c_S^2$ will decrease. In the limit of infinite mobility, $c_S^2$ reduces to $Var[V(1)]/\mathbb{E}[V(1)]^2$.\\
	For our computations, we assume the queue evolves at time slots of unit length, i.e., $\delta = 1$. The arrival process is a Bernoulli process of rate $p_{arr}$, where at every time slot a packet arrives with probability $p_{arr}$. The service process is $s(t) = \mathbbm{1}[\text{SINR}_0(t)>T]$, and we set $\Lambda = 0.1$, noise power $\sigma = 0$,  $R=0.3$, $T=8$ and $l(r) = (1+r)^{-4}$. We consider Rayleigh fades with mean 1 and a random direction mobility model. The parameter $p_{arr}$ is determined by the value of $\rho$, which we choose to be $0.97$. Under these assumptions and using the Fubini-Tonelli theorem, we can show that 
	\begin{align}
	\mathbb{E}[A(1)] &= p_{arr}, Var[A(1)] = p_{arr}(1-p_{arr})\\
	\mathbb{E}[V(1)] &= \mathbb{P}[\text{SINR}_0(t)>T] = \mathbb{P}(L_t) = \exp\left(-2\pi\Lambda \int_0^\infty \frac{u}{1 + l(R)/(Tl(u))}du\right)\\	
	Cov(V(1), V(j)) &= \int_0^1 \int_0^1 \mathbb{P}[\text{SINR}_0(t)>T, \text{SINR}_0(s + j - 1)>T] ds dt - \mathbb{P}(L_t)^2\label{eq:cov},
	\end{align}
	where the last quantity can be computed using methods similar to those presented in the Appendix. For the chosen parameters, Fig. \ref{fig:heavy_traffic_fig} compares mean workloads obtained via simulation (with 95\% confidence intervals) and via this heavy-traffic approximation. We see that for values of $\rho$ close to $1$, the approximation is accurate and provides a way to estimate the quantitative effect of velocity on mean workload.
		\begin{figure}[t]
		\begin{center}
			{\includegraphics[width=0.6\textwidth]{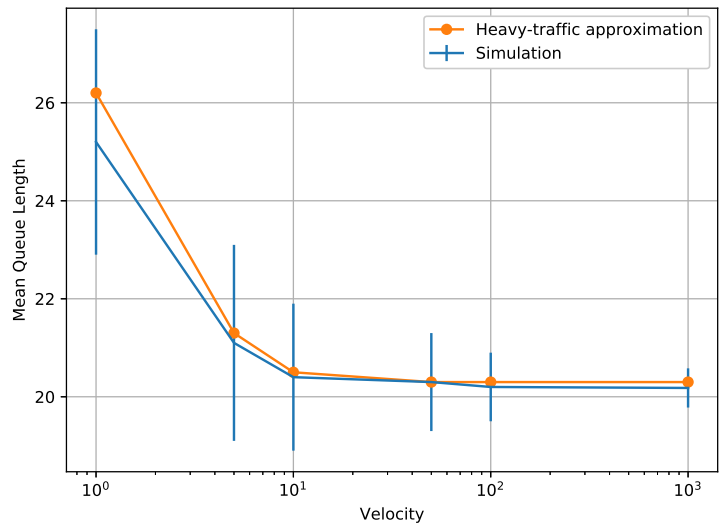}}
			\caption{Heavy-traffic approximation. $\rho = 0.97$.} 
			\label{fig:heavy_traffic_fig}
		\end{center}
	\end{figure}
	\subsection{Comparison of Mobility Models}
	Finally, we provide some insight into the ordering effect amongst mobility models that we observe in Fig. \ref{fig:mobility_comparison_fig}, drawing on results and intuition from Sections \ref{sec:correlation} and \ref{ss:heavy_traffic}. First note that the joint probabilities of SINR level-crossing events that we discussed in Section \ref{sec:correlation} appear in the heavy-traffic approximations that we presented in \ref{ss:heavy_traffic} via Eq. (\ref{eq:c_s_2}) and (\ref{eq:cov}). It is intuitive that these joint probabilities/correlation structures will also drive the ordering of mean workloads of queues that are served by any general rate function and that are in non-heavy traffic regimes - as we mention in Section \ref{sec:correlation}, decreased correlations across time slots should intuitively lead to smaller mean queue workloads. It seems natural, then, to compare these correlation structures across mobility models. We compute the quantity $\mathbb{P}[\text{SINR}_0(t)>T, \text{SINR}_0(t+1)>T] =  \mathbb{P}(L_t, L_{t+1})$, which we derive closed-form expressions for in the Appendix. We assume the same system parameters as in the previous subsections, i.e., $\Lambda = 0.1$, noise power $\sigma = 0$,  $R=0.3$, $T=8$ and $l(r) = (1+r)^{-4}$, and plot $\mathbb{P}(L_t, L_{t+1})$ as a function of velocity in Fig. \ref{fig:cov_comparison}.
			\begin{figure}[t]
		\begin{center}
			{\includegraphics[width=0.6\textwidth]{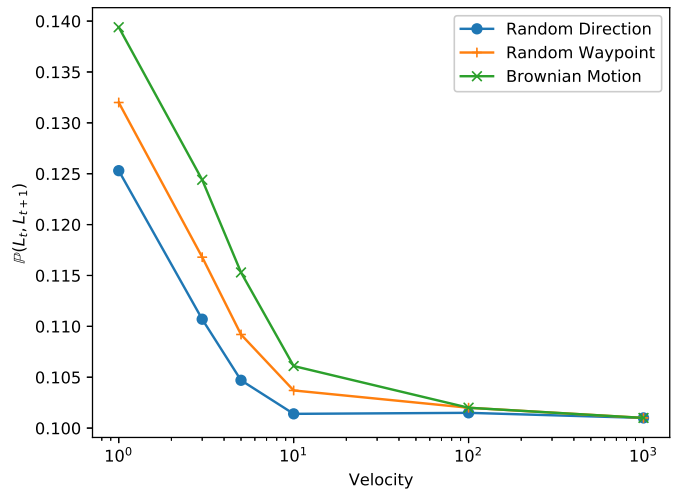}}
			\caption{The ordering of $\mathbb{P}(L_t, L_{t+1})$ across mobility models.} 
			\label{fig:cov_comparison}
		\end{center}
	\end{figure}
We see that the joint level-crossing probabilities are ordered across mobility models in the same way as the workloads in Fig. \ref{fig:mobility_comparison_fig}, confirming the intuition we presented in Section \ref{ss:avg_workload}.

	\section{Conclusion}
	We have shown that introducing mobility in a wireless network introduces mixing diversity that breaks non-ergodicity and allows us to provide universal stability guarantees for a queue driven by the resulting interference process. We next investigated the continuum of mobility scenarios that arise from considering all velocities in the interval $(0, \infty)$. Accelerating the motion of interferers accelerates the interference process, and the resulting averaging effect causes less variable and more reliable service that improves mean queue workload and mean delay. This effect can also be explained by considering the structure of correlations of the SINR and interference processes, which we have formalized. We also highlighted how our results and analysis \textcolor{black}{can be used to predict the behaviour of the more general and complex interacting queues setting}. Simulation and analytic studies (based on heavy-traffic approximations) provided further understanding and additional system design insights.
	\bibliographystyle{IEEEtran}  
	\nocite{*}
	\bibliography{references}
	 \appendix
	 \section{Proof of Theorem \ref{positive_correlation_crossing}}
	 We consider level-crossings spaced unit time apart for ease of exposition, but the result will hold for any interval of time. That $\mathbb{P}(L_{t+1}) = \mathbb{P}(L_t)$ is clear from the stationarity of $\text{SINR}_0(t)$. Let $h_1$ and $h_2$ be independent fading variables. The results in this section hold for any motion that preserves the homogeneity of $\Phi_t$. 
	 \vspace{-0.2cm}
	 \begin{align*}
	 &\mathbb{P}[L_{t+1}, L_t] = \mathbb{P}[l(R) h_1 > T I(t), l(R) h_2 > T I(t+1)]\\
	 &= \mathbb{E}[\exp(-\mu T I(t)/l(R)) \exp(-\mu T I(t+1)/l(R))]=\mathcal{L}_{t, t+1}(s, s),
	 \end{align*}
	 where $s = \frac{\mu T}{l(R)}.$
	 We assume that the motion that transforms $\Phi_t$ to $\Phi_{t+1}$ can be represented by a probability kernel $p$, where $p(x, B)$ is the probability that a point at location $x$ at time $t$ moves to region $B$ at time $t+1$. Denote by $\mathcal{L}_h$ the Laplace transform of the general fading variables $h$ that govern the channel between inteferers and the receiver of interest.
	 \begin{flalign*}
	 \mathcal{L}_{t, t+1}(s_1, s_2) &=\mathbb{E}\left[\prod_{X_i \in \Phi_t} \mathcal{L}_h (s_1 l(X_i))\prod_{Y_j \in \Phi_{t+1}} \mathcal{L}_h (s_2 l(Y_j))\right] \nonumber&\\
	 &=\mathbb{E}\left[\prod_{X_i \in \Phi_t} \mathcal{L}_h (s_1 l(X_i)) \int \mathcal{L}_h (s_2 l(y_i))p(X_i, dy_i)\right]\nonumber \\ 
	 &= \exp\left(-\Lambda \int (1 - v(x)) dx)\right), \nonumber
	 \end{flalign*}
	 \begin{align*}
	 \text{where }v(x) &= \mathcal{L}_h (s_1 l(x)) \int \mathcal{L}_h (s_2 l(y))p(x, dy) \\& = \mathcal{L}_h (s_1 l(x)) + \mathcal{L}_h (s_1 l(x)) \left[ \int \left( \mathcal{L}_h (s_2 l(y)) - 1 \right)p(x, dy)\right].
	 \end{align*}
	 Setting $s_1 = s_2 = s = \frac{\mu T}{l(R)}$, we have that $\mathcal{L}_{t, t+1}(s, s) = \mathbb{P}(L_t, L_{t+1})$. \\Noting also that $\mathbb{P}(L_t) = \exp\left(-\Lambda \int   1 - \mathcal{L}_h (s l(x))  dx\right)$, we obtain:
	 \begin{align*}
	 &\mathbb{P}(L_{t+1}|L_t) = \frac{\mathbb{P}(L_t, L_{t+1})}{\mathbb{P}(L_t)} = \frac{\mathcal{L}_{t, t+1}(s, s)}{\mathbb{P}(L_t)}  \\
	 &= \exp \left(\Lambda \int \int \mathcal{L}_h (sl(x)) \mathcal{L}_h(sl(y)) p(x, dy) dx \right) \times \exp \left(-\Lambda \int \mathcal{L}_h(sl(x))dx  \right).\\
	 &\implies \frac{\mathbb{P}(L_{t+1}|L_t)}{\mathbb{P}(L_t)} = \exp \Bigg(\Lambda \int \Big[ 1 - 2\mathcal{L}_h(sl(x))  + \int \mathcal{L}_h(sl(x)) \mathcal{L}_h (sl(y)) p(x, dy)   \Big] dx \Bigg).
	 \end{align*}
	 Since the system is stationary, $\mathbb{P}(L_{t+1})$ = $\mathbb{P}(L_t)$, which yields 
	 \begin{align*}
	 \int \left( \int \mathcal{L}_h (sl(y)) p(x, dy) \right) dx = \int \mathcal{L}_h(sl(x))  dx.
	 \end{align*}
	 {Plugging this last equation back in,}
	 \begin{align*}
	 & \frac{\mathbb{P}(L_{t+1}|L_t)}{\mathbb{P}(L_t)} = \exp \Bigg( \Lambda \int \left[1 - \mathcal{L}_h(sl(x))\right] \times \Bigg[1 - \Big(\int \mathcal{L}_h(sl(y)) p(x, dy)\Big)\Bigg] dx \Bigg) > 1.\qquad\square
	 \end{align*}
\end{document}